\documentclass[3p, numbers]{elsarticle}

\usepackage{hyperref}

\makeatletter
\def\ps@pprintTitle{%
	\let\@oddhead\@empty
	\let\@evenhead\@empty
	\def\@oddfoot{\centerline{\thepage}}%
	\let\@evenfoot\@oddfoot}
\makeatother

\usepackage{amsmath, amsthm, amssymb}

\theoremstyle{plain}
\newtheorem{theorem}{Theorem}

\newtheorem{lemma}[theorem]{Lemma}

\theoremstyle{definition}

\usepackage{graphicx, color}

\usepackage{algorithm, algpseudocode} 
\usepackage{mathrsfs} 

\newcommand{\msize}[1]{{\left|#1\right|}}
\newcommand{\Nei}[2]{N_{#1}(#2)}
\newcommand{\Neiclosed}[2]{N_{#1}[#2]}
\newcommand{\dist}{\mathsf{dist}}
\newcommand{\rev}{\mathsf{rev}}
\newcommand{\cost}{\mathsf{cost}}
\newcommand{\len}{\mathsf{len}}

\newcommand{\sfPSPACE}{\mathsf{PSPACE}}
\newcommand{\sfNP}{\mathsf{NP}}
\newcommand{\sfP}{\mathsf{P}}

\newcommand{\calF}{\mathcal{F}}
\newcommand{\calS}{S}
\newcommand{\calSp}{{\calS^\prime}}

\newcommand{\sfTS}{\mathsf{TS}}

\newcommand{\Ipp}{{I^{\prime\prime}}}
\newcommand{\Ip}{{I^{\prime}}}
\newcommand{\Jp}{{J^{\prime}}}

\newcommand{\Mstar}{{M^*}}
\newcommand{\Dstar}{{D^*}}

\newcommand{\sevstepT}[1]{\overset{#1}{\leftrightsquigarrow}}

\begin{document}
	
	\begin{frontmatter}
		
		\title{Shortest Reconfiguration Sequence for Sliding Tokens on Spiders}
		
		\author[jaist]{Duc~A.~Hoang}
		\ead{hoanganhduc@jaist.ac.jp}
		\ead[url]{http://hoanganhduc.github.io/}
		
		\author[kurdistan]{Amanj~Khorramian}
		\ead{khorramian@gmail.com}
		
		\author[jaist]{Ryuhei~Uehara\fnref{uehara}}
		\ead{uehara@jaist.ac.jp}
		\ead[url]{http://www.jaist.ac.jp/~uehara/}
		
		\address[jaist]{School of Information Science, JAIST,\\
			1-1 Asahidai, Nomi, Ishikawa, 923-1292 Japan}
		\address[kurdistan]{Department of Electrical and Computer Engineering,\\
			University of Kurdistan, Sanandaj, Iran}
		
		\fntext[uehara]{JSPS KAKENHI Grant
			Number JP17H06287 and 18H04091}
		
		\begin{abstract}
			Suppose that two independent sets $I$ and $J$ of a graph with $\vert I \vert = \vert J \vert$ are given,
			and a token is placed on each vertex in $I$. 
			The {\sc Sliding Token} problem is to determine 
			whether there exists a sequence of independent sets which transforms $I$ 
			into $J$ so that each independent set in the sequence results from 
			the previous one by sliding exactly one token along an edge in the graph.
			It is one of the representative reconfiguration problems 
			that attract the attention from the viewpoint of theoretical computer science.
			For a yes-instance of a reconfiguration problem, 
			finding a shortest reconfiguration sequence has a different aspect.
			In general, even if it is polynomial time solvable to decide whether two instances 
			are reconfigured with each other, it can be $\mathsf{NP}$-hard to find a shortest sequence between them.
			In this paper, we show that the problem for finding a shortest sequence between two independent sets 
			is polynomial time solvable for spiders (i.e., trees having exactly one vertex of degree at least three).
		\end{abstract}
		
		
	\end{frontmatter}
	
	\section{Introduction}
	\label{sec:intro}
	
	Recently, the \emph{reconfiguration problems} attracted the attention 
	from the viewpoint of theoretical computer science.
	These problem arise when we like to find a step-by-step transformation between 
	two feasible solutions of a problem such that all intermediate results are also feasible and 
	each step abides by a fixed reconfiguration rule, that is, 
	an adjacency relation defined on feasible solutions of the original problem.
	The reconfiguration problems have been studied extensively for several well-known problems, including 
	{\sc Independent Set}~\cite{HearnDemaine2005,IDHPSUU,KaminskiMedvedevMilanic2012,MNRSS13}, 
	{\sc Satisfiability}~\cite{Kolaitis,MTY11}, {\sc Set Cover}, {\sc Clique}, {\sc Matching}~\cite{IDHPSUU}, and so on.
	
	A reconfiguration problem can be seen as a natural ``puzzle'' from the viewpoint of 
	recreational mathematics. 
	The \emph{15-puzzle} is one of the most famous classic puzzles,
	that had the greatest impact on American and European societies (see \cite{Slocum} for its rich history). 
	It is well known that the 15-puzzle has a parity, and 
	one can solve the problem in linear time just by checking whether the parity of 
	one placement coincides with the other or not.
	Moreover, the distance between any two reconfigurable placements is $O(n^3)$, 
	that is, we can reconfigure from one to the other in $O(n^3)$ sliding pieces when the size of the board is $n\times n$.
	However, surprisingly, for these two reconfigurable placements, 
	finding a shortest path is $\sfNP$-complete in general \cite{DR2017,RW90}.
	Namely, although we know that there is a path of length in $O(n^3)$, finding a shortest one is $\sfNP$-complete.
	While every piece is a unit square in the 15-puzzle, we obtain 
	the other famous classic puzzle when we allow to have rectangular pieces,
	which is called ``Dad puzzle'' and its variants can be found in the whole world
	(e.g., it is called ``hako-iri-musume'' in Japanese).
	Gardner said that ``these puzzles are very much in want of a theory'' in 1964 \cite{Gardner},
	and Hearn and Demaine gave the theory after 40 years \cite{HearnDemaine2005};
	they are $\sfPSPACE$-complete in general \cite{HearnDemaine2009}.
	
	Summarizing up, these sliding block puzzles characterize representative computational complexity classes;
	the decision problem for unit squares can be solved in linear time just by checking parities,
	finding a shortest reconfiguration for the unit squares is $\sfNP$-complete, and 
	the decision problem becomes $\sfPSPACE$-complete for rectangular pieces.
	That is, this simple reconfiguration problem gives us a new sight of these 
	representative computational complexity classes.
	
	In general, the reconfiguration problems tend to be $\sfPSPACE$-complete,
	and some polynomial time algorithms are shown in restricted cases.
	Finding a shortest sequence in the context of the reconfiguration problems 
	is a new trend in theoretical computer science because it has 
	a great potential to characterize the class $\sfNP$ from 
	a different viewpoint from the classic ones.
	
	\begin{figure}[!ht]
		\centering
		\includegraphics[scale=0.7]{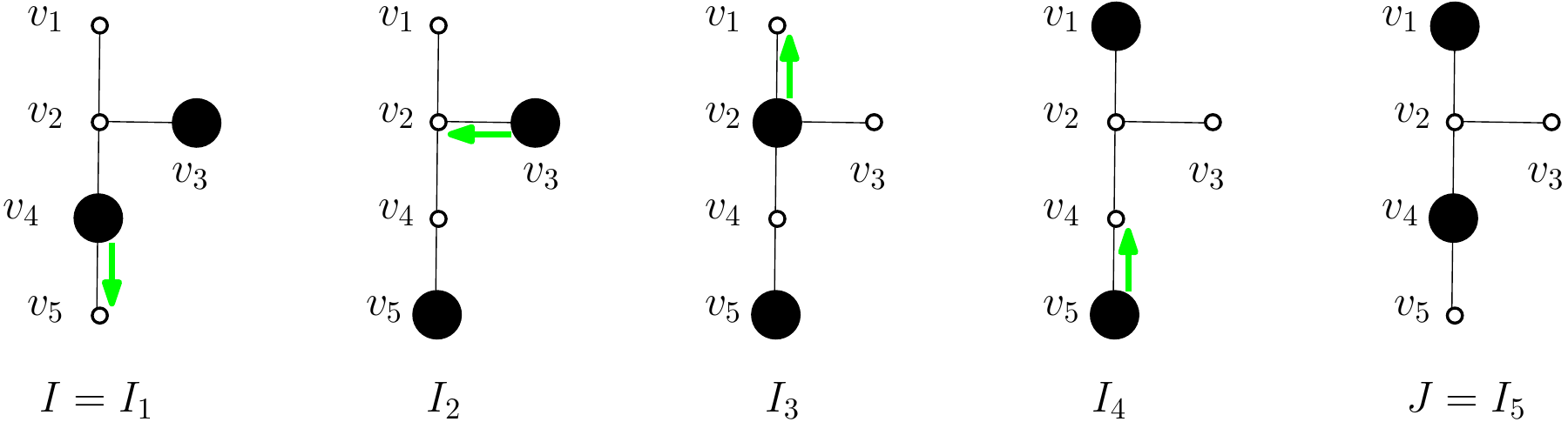}
		\caption{A sequence $\langle I_1, I_2, \ldots, I_5 \rangle$ of independent sets of the same graph,
			where the vertices in independent sets are depicted by small black circles (tokens).}
		\label{fig:example}
	\end{figure}
	
	One of the important $\sfNP$-complete problems is 
	the {\sc Independent Set} problem.
	For this notion, a natural reconfiguration problem called {\sc Sliding Token} was introduced by Hearn and Demaine~\cite{HearnDemaine2005}.
	(See~\cite{KaminskiMedvedevMilanic2012} for an overview on different reconfiguration variants of {\sc Independent Set}.)
	Suppose that we are given two independent sets  $I$ and $J$ of 
	a graph $G=(V,E)$ such that $\msize{I}=\msize{J}$, 
	and imagine that a \emph{token} (coin) is placed on each vertex in $I$. 
	For convenience, sometimes we identify the token with the vertex it is placed on and simply say ``a token in an independent set.''
	Then, the {\sc Sliding Token} problem is to determine whether there exists a sequence 
	$\calS = \langle I_1, I_2, \ldots, I_{\ell} \rangle$ of independent sets of $G$ such that
	\begin{itemize}
		\item[(a)] $I_1=I$, $I_{\ell}=J$, 
		and $\msize{I_i} = \msize{I}=\msize{J}$ for all $i$, $1 \le i \le \ell$; and 
		\item[(b)] for each $i$, $2 \le i \le \ell$, 
		there is an edge $xy$ in $G$ such that $I_{i-1} \setminus I_{i}=\{x\}$ 
		and $I_{i}\setminus I_{i-1}=\{y\}$.
	\end{itemize}
	That is, $I_{i}$ can be obtained from $I_{i-1}$ by sliding exactly 
	one token on a vertex $x \in I_{i-1}$ to its adjacent vertex $y \in I_{i}$ 
	along an edge $xy \in E(G)$.
	Such a sequence $\calS$, if exists, is called a $\sfTS$-\emph{sequence} in $G$ between $I$ and $J$.
	We denote by a $3$-tuple $(G, I, J)$ an instance of {\sc Sliding Token} problem.
	If a $\sfTS$-sequence $\calS$ in $G$ between $I$ and $J$ exists, we say that $I$ is \emph{reconfigurable} to $J$ (and vice versa), and write $I \sevstepT{G} J$.
	The sets $I$ and $J$ are the \emph{initial} and \emph{target} independent sets, respectively.
	For a $\sfTS$-sequence $\calS$, the \emph{length} $\len(\calS)$ of $\calS$ is defined as the number of independent sets in $\calS$ minus one.
	In other words, $\len(\calS)$ is the number of token-slides described in $\calS$.
	Figure~\ref{fig:example} illustrates a $\sfTS$-sequence of length $4$ between two independent sets $I = I_1$ and $J = I_5$.
	
	For the {\sc Sliding Token} problem, linear-time algorithms have been shown for cographs 
	(also known as $P_4$-free graphs) \cite{KaminskiMedvedevMilanic2012} and trees \cite{DDFEHIOOUY2015}. 
	Polynomial-time algorithms are shown for bipartite permutation graphs \cite{FoxEpsteinHoangOtachiUehara2015}, 
	claw-free graphs \cite{BonsmaKaminskiWrochna}, cacti \cite{HoangUehara2016}, and interval graphs \cite{BonamyB17} \footnote{
		We note that the algorithm for a block graph in \cite{HoangFoxEpsteinUehara2017} has a flaw,
		and hence it is not yet settled \cite{ViniCom2018}.}.
	On the other hand, $\sfPSPACE$-completeness is also shown for 
	graphs of bounded tree-width \cite{MouawadNishimuraRamanWrochna}, 
	planar graphs \cite{HearnDemaine2005,HearnDemaine2009}, planar graphs with bounded bandwidth \cite{Zanden}, and split graphs~\cite{BelmonteKLMOS18}.
	
	In this context, for a given {\sc yes}-instance $(G, I, J)$ of {\sc Sliding Token}, we aim to find a shortest $\sfTS$-sequence between $I$ and $J$.
	Such a problem is called the {\sc Shortest Sliding Token} problem.
	As seen for the 15-puzzle, the {\sc Shortest Sliding Token} problem can be intractable 
	even for these graph classes which the decision problem can be solved in polynomial time.
	Moreover, in the 15-puzzle, we already know that it has a solution of polynomial length for two configurations.
	However, in the {\sc Sliding Token} problem, we have no upper bound of the length of a solution in general.
	To deal with this delicate issue, we have to distinguish two variants of this problem.
	In the {\em decision variant}, an integer $\ell$ is also given as a part of input, 
	and we have to decide whether there exists a sequence between $I$ and $J$ of length at most $\ell$. 
	In the {\em non-decision variant}, we are asked to output a specific shortest $\sfTS$-sequence.
	The length $\ell$ is not necessarily polynomial in $\msize{V(G)}$ in general.
	When $\ell$ is super-polynomial, we may have that the decision variant is in $\sfP$, 
	while the non-decision one is not in $\sfP$ since it takes super-polynomial time to output the sequence.
	On the other hand, even when $G$ is a perfect graph and $\ell$ is polynomial in $\msize{V(G)}$, the decision variant of {\sc Shortest Sliding Token} is $\sfNP$-complete (see~\cite[Theorem 5]{KaminskiMedvedevMilanic2012}).
	In short, in the decision variant, we focus on the {\em length} of a shortest $\sfTS$-sequence, while in the non-decision variant, we focus on the {\em construction} of a shortest $\sfTS$-sequence itself.
	
	From this viewpoint, the length of a token sliding is a key feature of the 
	{\sc Shortest Sliding Token} problem. If the length is super-polynomial in total, 
	there exists at least one token that slides super-polynomial times. 
	That is, the token visits the same vertex many times in its slides.
	That is, some tokens make \emph{detours} in the sequence 
	(the notion of detour is important and precisely defined later).
	In general, it seems to be more difficult to analyze ``detours of tokens'' for graphs containing cycle(s).
	As a result, one may first consider the problem for trees.
	The {\sc Sliding Token} problem on a tree 
	can be solved in linear time \cite{DDFEHIOOUY2015}. 
	Polynomial-time algorithms for the {\sc Shortest Sliding Token} problem were first investigated in
	\cite{YamadaUehara2016}. 
	In~\cite{YamadaUehara2016}, the authors gave polynomial-time algorithms for solving {\sc Shortest Sliding Token} when the input graph is either a proper interval graph, a trivially perfect graph, or a caterpillar.
	We note that caterpillars is the first graph class that required detours to solve 
	the {\sc Shortest Sliding Token} problem.
	A caterpillar is a tree that consists of a ``backbone'' called a \emph{spine} with many \emph{pendants}, 
	or leaves attached to the spine. Each pendant can be used to escape a token, 
	however, the other tokens cannot pass through it. Therefore, the ordering of tokens on the spine is fixed.
	In this paper, we consider the {\sc Shortest Sliding Token} problem on a spider, 
	which is a tree with one central vertex of degree more than $2$. 
	On this graph, we can use each ``leg'' as a stack and exchange tokens using these stacks.
	Therefore, we have many ways to handle the tokens, and hence we need more analyses to find a shortest sequence.
	In this paper, we give an $O(n^2)$ time algorithms for the {\sc Shortest Sliding Token} problem on a spider,
	where $n$ is the number of vertices.
	The algorithm is constructive, and the sequence itself can be output in $O(n^2)$ time.
	As mentioned in \cite{YamadaUehara2016}, the number of required token-slides in a sequence can be $\Omega(n^2)$,
	hence our algorithm is optimal for the number of token-slides.
	
	\noindent
	{\bf Note:} Recently, it is announced that the {\sc Shortest Sliding Token} problem on 
	a tree can be solved in polynomial time by Sugimori \cite{Sugimori2018a}.
	His algorithm is based on a dynamic programming on a tree \cite{Sugimori2018b}:
	though it runs in polynomial time, it seems to have much larger degree comparing 
	to our case-analysis based algorithm.
	
	\section{Preliminaries}
	\label{sec:preliminaries}
	
	For common graph theoretic definitions, we refer the readers to the textbook \cite{Diestel2010}.
	Throughout this paper, we denote by $V(G)$ and $E(G)$ the vertex-set and edge-set of a graph $G$, respectively.
	We always use $n$ for denoting $\msize{V(G)}$.
	For a vertex $x \in V(G)$, we denote by $\Nei{G}{x}$ the set $\{y \in V(G): xy \in E(G)\}$ of \emph{neighbors} of $x$, and by $\Neiclosed{G}{x}$ the set $\Nei{G}{x} \cup \{x\}$ of \emph{closed neighbors} of $x$.
	In a similar manner, for an induced subgraph $H$ of $G$, the set $\Neiclosed{G}{H}$ is defined as $\bigcup_{x \in V(H)}\Neiclosed{G}{x}$.
	The \emph{degree} of $x$, denoted by $\deg_G(x)$, is the size of $\Nei{G}{x}$. 
	For $x, y \in V(G)$, the \emph{distance} $\dist_G(x, y)$ between $x$ and $y$ is simply the length (i.e., the number of edges) of a shortest $xy$-path in $G$.
	
	For a tree $T$, we denote by $P_{xy}$ the (unique) shortest $xy$-path in $T$, and by $T^x_y$ the subtree of $T$ induced by $y$ and its descendants when regarding $T$ as the tree rooted at $x$.
	A \emph{spider graph} (or \emph{starlike tree}) is a tree having exactly one vertex (called its \emph{body}) of degree at least $3$.
	For a spider $G$ with body $v$ and a vertex $w \in \Nei{G}{v}$, the path $G^v_w$ is called a \emph{leg} of $G$.
	By definition, it is not hard to see that two different legs of $G$ have no common vertex.
	For example, the graph in Figure~\ref{fig:example} is a spider with body $v = v_2$ and $\deg_G(v) = 3$ legs attached to $v$.
	
	Let $(G, I, J)$ be an instance of {\sc Shortest Sliding Token}.
	A \emph{target assignment} from $I$ to $J$ is simply a bijective mapping $f: I \to J$.
	A target assignment $f$ is called \emph{proper} if there exists a $\sfTS$-sequence in $G$ between $I$ and $J$ that moves the token on $w$ to $f(w)$ for every $w \in I$.
	Given a target assignment $f: I \to J$ from $I$ to $J$, one can also define the target assignment $f^{-1}: J \to I$ from $J$ to $I$ as follows: for every $x \in J$, $f^{-1}(x) = \{y \in I: f(y) = x\}$.
	Let $\calF$ be the set of all target assignments from $I$ to $J$.
	We define $\Mstar(G, I, J) = \min_{f \in \calF}\sum_{w \in I}\dist_G(w, f(w))$.
	Intuitively, observe that any $\sfTS$-sequence between $I$ and $J$ in $G$ (if exists) uses at least $\Mstar(G, I, J)$ token-slides.
	
	Let $\calS = \langle I_1, I_2, \dotsc, I_\ell \rangle$ be a $\sfTS$-sequence between two independent sets $I = I_1$ and $J = I_\ell$ of a graph $G$.
	Indeed, one can describe $\calS$ in term of token-slides as follows: $\calS = \langle x_1 \to y_1, x_2 \to y_2, \dotsc, x_{\ell - 1} \to y_{\ell - 1} \rangle$, where $x_i$ and $y_i$ ($i \in \{1, 2, \dotsc, \ell-1\}$) satisfy $x_iy_i \in E(G)$, $I_i \setminus I_{i+1} = \{x_i\}$, and $I_{i+1} \setminus I_i = \{y_i\}$.
	The \emph{reverse} of $\calS$ (which reconfigures $J$ to $I$), denoted by $\rev(\calS)$, is defined by $\rev(\calS) = \langle I_\ell, \dotsc, I_2, I_1 \rangle$.
	One can also describe $\rev(\calS)$ in term of token-slides: $\rev(\calS) = \langle y_{\ell-1} \to x_{\ell-1}, \dots, y_2 \to x_2, y_1 \to x_1 \rangle$.
	For example, the $\sfTS$-sequence $\calS = \langle I_1, \dotsc, I_5 \rangle$ described in Figure~\ref{fig:example} can also be written as $\calS = \langle v_4 \to v_5, v_3 \to v_2, v_2 \to v_1, v_5 \to v_4 \rangle$.
	Similarly, $\rev(\calS) = \langle I_5, \dotsc, I_1 \rangle = \langle v_4 \to v_5, v_1 \to v_2, v_2 \to v_3, v_5 \to v_4 \rangle$.
	
	For an edge $e = xy \in E(G)$, we say that $\calS$ \emph{makes detour over} $e$ if both $x \to y$ and $y \to x$ are members of $\calS$.
	We emphasize that the steps $x \to y$ and $y \to x$ is {\em not} necessarily made by the same token.
	The \emph{number of detours $\calS$ makes over $e$}, denoted by $D_G(\calS, e)$, is defined to be twice the minimum between the number of appearances of $x \to y$ and the number of appearances of $y \to x$.
	The \emph{total number of detours $\calS$ makes in $G$}, denoted by $D_G(\calS)$, is defined to be $\sum_{e \in E(G)}D_G(\calS, e)$.
	As an example, one can verify that the $\sfTS$-sequence $\calS$ described in Figure~\ref{fig:example} satisfies $D_G(\calS, v_4v_5) = 2$ and $D_G(\calS) = 2$.
	Let $\mathcal{S}$ be the set of all $\sfTS$-sequences in $G$ between two independent sets $I, J$.
	We define by $\Dstar(G, I, J) = \min_{\calS \in \mathcal{S}}D_G(\calS)$ the smallest number of detours that a $\sfTS$-sequence between $I$ and $J$ in $G$ can possibly make.
	
	For two $\sfTS$-sequences $\calS_1 = \langle x_1 \to y_1, x_2 \to y_2, \dotsc, x_{\ell - 1} \to y_{\ell - 1} \rangle$ and $\calS_2 = \langle x_1^\prime \to y_1^\prime, x_2^\prime \to y_2^\prime, \dotsc, x_p^\prime \to y_p^\prime \rangle$ in a graph $G$, if the sequence of token-slides $\calS = \langle x_1 \to y_1, x_2 \to y_2, \dotsc, x_{\ell - 1} \to y_{\ell - 1}, x_1^\prime \to y_1^\prime, x_2^\prime \to y_2^\prime, \dotsc, x_p^\prime \to y_p^\prime \rangle$ forms a $\sfTS$-sequence in $G$, we define $\calS = \calS_1 \oplus \calS_2$, and say that $\calS$ is obtained by taking the \emph{concatenation} of $\calS_1$ and $\calS_2$.
	
	\section{{\sc Shortest Sliding Token} for spiders}
	\label{sec:shortest-TS-seq-spiders}
	
	In this section, we show that {\sc Shortest Sliding Token} for spiders can be solved in polynomial time.
	More precisely, we claim that
	\begin{theorem}\label{thrm:main}
		Given an instance $(G, I, J)$ of {\sc Shortest Sliding Token} for spiders, one can construct a shortest $\sfTS$-sequence between $I$ and $J$ in $O(n^2)$ time, where $n$ denotes the number of vertices of the given spider $G$.
	\end{theorem}
	First of all, from the linear-time algorithm for solving {\sc Sliding Token} for trees (which also applies for spiders as well) presented in~\cite{DDFEHIOOUY2015}, we can simplify our problem as follows.
	For an independent set $I$ of a tree $T$, the token on $u \in I$ is said to be $(T, I)$-\emph{rigid} if for any $\Ip$ with $I \sevstepT{T} \Ip$, $u \in \Ip$.
	Intuitively, a $(T, I)$-rigid token cannot be moved by any $\sfTS$-sequence in $T$. 
	One can find all $(T, I)$-rigid tokens in a given tree $T$ in linear time.
	Moreover, a $\sfTS$-sequence between $I$ and $J$ in $T$ exists if and only if the $(T, I)$-rigid tokens and $(T, J)$-rigid tokens are the same, and for any component $F$ of the forest obtained from $T$ by removing all vertices where $(T, I)$-rigid tokens are placed and their neighbors, $\msize{I \cap F} = \msize{J \cap F}$.
	Thus, for an instance $(G, I, J)$ of {\sc Shortest Sliding Token} for spiders, we can assume without loss of generality that $I \sevstepT{G} J$ and there are no $(G, I)$-rigid and $(G, J)$-rigid tokens.
	
	\subsection{General idea}
	\label{sec:general-idea}
	We now give a brief overview of our approach.
	For convenience, from now on, let $(G, I, J)$ be an instance of {\sc Shortest Sliding Token} for spiders satisfying the above assumption.
	Rough speaking, we aim to construct a $\sfTS$-sequence in $G$ between $I$ and $J$ of {\em minimum} length $\Mstar(G, I, J) + \Dstar(G, I, J)$, where $\Mstar(G, I, J)$ and $\Dstar(G, I, J)$ are respectively the smallest number of token-slides and the smallest number of detours that a $\sfTS$-sequence between $I$ and $J$ in $G$ can possibly perform, as defined in the previous section. 
	Indeed, the following lemma implies that any $\sfTS$-sequence in $G$ between $I$ and $J$ must be of length at least $\Mstar(G, I, J) + \Dstar(G, I, J)$.
	
	\begin{lemma}
		\label{lem:lower-bound-len-TS-seq-tree}
		Let $I, J$ be two independent sets of a tree $T$ such that $I \sevstepT{T} J$.
		Then, for every $\sfTS$-sequence $\calS$ between $I$ and $J$, $\len(\calS) \geq \Mstar(T, I, J) + \Dstar(T, I, J)$.
	\end{lemma}
	
	\begin{proof}
		Let $I = \{w_1, w_2, \dotsc, w_\msize{I}\}$.
		Let $\calS$ be a $\sfTS$-sequence between $I$ and $J$ that moves the token $t_i$ on $w_i$ to $f(w_i)$ for some target assignment $f: I \to J$.
		For each $i \in \{1, 2, \dotsc, \msize{I}\}$, let $\calS_i$ be the sequence of $\dist_T(w_i, f(w_i))$ token-slides that moves $t_i$ from $w_i$ to $f(w_i)$ along the (unique) path $P_{w_if(w_i)}$.
		Note that $\calS_i$ is not necessarily a $\sfTS$-sequence.
		
		Let consider the movements of $t_i$ from $w_i$ to $f(w_i)$ in the $\sfTS$-sequence $\calS$. First of all, it is clear that $t_i$ needs to make all moves in $\calS_i$.
		Since the path $P_{w_if(w_i)}$ is unique, if $t_i$ makes any move $x \to y$ that is not in $\calS_i$ for some edge $xy \in E(T)$, it must also make the move $y \to x$ later, hence forming detour over $e$.
		Let $D_1$ be the number of detours formed by the token-slides in $\calS \setminus \bigcup_{i=1}^{\msize{I}}\calS_i$.
		Clearly, $\len(\calS) = \sum_{i=1}^{\msize{I}}\dist_T(w_i, f(w_i)) + D_1$.
		
		The token-slides in $\bigcup_{i=1}^{\msize{I}}\calS_i$ may also form detour. 
		Let $i, j \in \{1, 2, \dotsc, \msize{I}\}$ be such that the sequence $\calS_i$ moves $t_i$ from $w_i$ to $f(w_i)$ and at some point makes the move $x \to y$, and the sequence $\calS_j$ moves $t_j$ from $w_j$ to $f(w_j)$ and at some point makes the move $y \to x$.
		Together, $\calS_i$ and $\calS_j$ form detour over an edge $e = xy \in E(P_{w_if(w_i)}) \cap E(P_{w_jf(w_j)})$.
		Let $D_2$ be the number of detours formed by such token-slides.
		Clearly, $D_G(\calS) = D_1 + D_2$.
		
		Suppose that for an edge $e = xy \in E(T)$, there exists $k_e$ pairs $(i_1, j_1), (i_2, j_2), \dotsc, (i_{k_e}, j_{k_e})$ with $1 \leq i_p, j_p \leq \msize{I}$, $i_p \neq j_p$, and for any two pairs $(i_p, j_p)$ and $(i_q, j_q)$, $i_p \neq i_q$ and $j_p \neq j_q$ ($1 \leq p, q \leq k_e$) such that for each $p \in \{1, 2, \dotsc, k_e\}$, the sequence $\calS_{i_p}$ at some point makes the move $x \to y$, and the sequence $\calS_{j_p}$ at some point makes the move $y \to x$.
		It follows that the vertices $\{w_{i_p}\}_{1 \leq p \leq k_e}$ and $\{f(w_{j_p})\}_{1 \leq p \leq k_e}$ are in $V(T^y_x)$, 
		and the vertices $\{w_{j_p}\}_{1 \leq p \leq k_e}$ and $\{f(w_{i_p})\}_{1 \leq p \leq k_e}$ are in $V(T^x_y)$.
		We note that $1 \leq k_e \leq \lfloor \msize{I}/2 \rfloor$, and emphasize again that $\calS_{i_p}$ and $\calS_{j_p}$ are not necessarily $\sfTS$-sequences.
		Let $\mathcal{E}_f$ be the set of all edges of $T$ satisfying the described property with respect to the target assignment $f$.
		Then, $D_2 = 2\sum_{e \in \mathcal{E}_f}k_e$.
		
		Let $e \in \mathcal{E}_f$ be an edge of $T$ as described above. 
		Let $g$ be the target assignment defined as follows: for $1 \leq p \leq k_e$, $g(w_{i_p}) = f(w_{j_p})$, $g(w_{j_p}) = f(w_{i_p})$, and $g(w_i) = f(w_i)$ for $i \notin \{i_1, i_2, \dotsc, i_{k_e}, j_1, j_2, \dotsc, j_{k_e}\}$.
		Then, $\sum_{i=1}^{\msize{I}}\dist_T(w_i, f(w_i)) = \sum_{i=1}^{\msize{I}}\dist_T(w_i, g(w_i)) + 2k_e$, and $\mathcal{E}_g = \mathcal{E}_f \setminus \{e\}$.
		Using this property repeatedly, we can finally find a target assignment $g$ such that $\mathcal{E}_g = \emptyset$ and $\sum_{i=1}^{\msize{I}}\dist_T(w_i, f(w_i)) = \sum_{i=1}^{\msize{I}}\dist_T(w_i, g(w_i)) + D_2$.
		
		Therefore, $\len(\calS) = \sum_{i=1}^{\msize{I}}\dist_T(w_i, f(w_i)) + D_1 = \sum_{i=1}^{\msize{I}}\dist_T(w_i, g(w_i)) + D_2 + D_1 \geq \Mstar(T, I, J) + \Dstar(T, I, J)$.
	\end{proof}
	
	As a result, it remains to show that any $\sfTS$-sequence in $G$ between $I$ and $J$ must be of length at most $\Mstar(G, I, J) + \Dstar(G, I, J)$, and there exists a specific $\sfTS$-sequence $\calS$ in $G$ between $I$ and $J$ whose length is exactly $\Mstar(G, I, J) + \Dstar(G, I, J)$.
	To this end, we shall analyze the following cases.
	\begin{itemize}
		\item {\bf Case~1:} $\max\{\msize{I \cap \Nei{G}{v}}, \msize{J \cap \Nei{G}{v}}\} = 0$.
		\item {\bf Case~2:} $0 < \max\{\msize{I \cap \Nei{G}{v}}, \msize{J \cap \Nei{G}{v}}\} \leq 1$.
		\item {\bf Case~3:} $\max\{\msize{I \cap \Nei{G}{v}}, \msize{J \cap \Nei{G}{v}}\} \geq 2$.
	\end{itemize}
	In each case, we claim that it is possible to simultaneously determine $\Dstar(G, I, J)$ and construct a $\sfTS$-sequence in $G$ between $I$ and $J$ whose length is minimum.
	More precisely, in {\bf Case~1}, we show that it is always possible to construct a $\sfTS$-sequence between $I$ and $J$ of length $\Mstar(G, I, J)$, that is, no detours are required. 
	(Note that, no $\sfTS$-sequence can use less than $\Mstar(G, I, J)$ token-slides.)
	However, this does not hold in {\bf Case~2}.
	In this case, we show that in certain conditions, detours cannot be avoided, that is, any $\sfTS$-sequence must make detours at least one time at some edge of $G$.
	More precisely, in such situations, we show that it is possible to construct a $\sfTS$-sequence between $I$ and $J$ of length $\Mstar(G, I, J) + 2$, that is, the sequence makes detour at exactly one edge. 
	Finally, in {\bf Case~3}, we show that detours cannot be avoided at all, and it is possible to construct a $\sfTS$-sequence between $I$ and $J$ of minimum length, without even knowing exactly how many detours it performs.
	As a by-product, we also describe how one can calculate this (smallest) number of detours precisely.
	
	\subsection{When $\max\{\msize{I \cap \Nei{G}{v}}, \msize{J \cap \Nei{G}{v}}\} = 0$}
	As mentioned before, in this case, we will describe how to construct a $\sfTS$-sequence $\calS$ in $G$ between $I$ and $J$ whose length $\len(\calS)$ equals $\Mstar(G, I, J) + \Dstar(G, I, J)$. 
	In general, to construct any $\sfTS$-sequence, we need: (1) a target assignment $f$ that tells us the final position a token should be moved to (say, a token on $v$ should finally be moved to $f(v)$); and (2) an ordering of tokens that tells us which token should move first. 
	From the definition of $\Mstar(G, I, J)$, it is natural to require that our target assignment $f$ satisfies $\Mstar(G, I, J) = \sum_{w \in I}\dist_G(w, f(w))$.
	As you will see later, such a target assignment exists, and we can always construct one in polynomial time.
	We also claim that one can efficiently define a total ordering $\prec$ of vertices in $I$ such that if $x, y \in I$ and $x \prec y$, then the token on $x$ will be moved before the token on $y$ in our desired $\sfTS$-sequence.
	Combining these results, our desired $\sfTS$-sequence will finally be constructed (in polynomial time).
	
	\subparagraph*{Target assignment.} 
	We now describe how to construct a target assignment $f$ such that $\Mstar(G, I, J) = \sum_{w \in I}\dist_G(w, f(w))$.
	For convenience, we always assume that the given spider $G$ has body $v$ and $\deg_G(v)$ legs $L_1, \dotsc, L_{\deg_G(v)}$. 
	Moreover, we assume without loss of generality that these legs are labeled such that $\msize{I \cap V(L_i)} - \msize{J \cap V(L_i)} \leq \msize{I \cap V(L_j)} - \msize{J \cap V(L_j)}$ for $1 \leq i \leq j \leq \deg_G(v)$; otherwise, we simply re-label them.
	For each leg $L_i$ ($i \in \{1, 2, \dotsc, \deg_G(v)\}$), we define the corresponding independent sets $I_{L_i}$ and $J_{L_i}$ as follows: $I_{L_1} = (I \cap V(L_1)) \cup (I \cap \{v\})$; $J_{L_1} = (J \cap V(L_1)) \cup (J \cap \{v\})$; and for $i \in \{2, \dotsc, d\}$, we define $I_{L_i} = I \cap V(L_i)$ and $J_{L_i} = J \cap V(L_i)$.
	In this way, we always have $v \in I_{L_1}$ (resp. $v \in J_{L_1}$) if $v \in I$ (resp. $v \in J$). 
	This definition will be helpful when considering tokens placed at the body vertex $v$.
	
	Under the above assumptions, we design Algorithm~\ref{algo:find-target-assignment} for constructing $f$ as below.
	\begin{algorithm}
		\caption{Find a target assignment between two independent sets $I, J$ of a spider $G$ such that $\Mstar(G, I, J) = \sum_{w \in I}\dist_G(w, f(w))$.}
		\label{algo:find-target-assignment}
		\textbf{Input:} Two independent sets $I, J$ of a spider $G$ with body $v$. \\
		\textbf{Output:} A target assignment $f: I \to J$ such that $\Mstar(G, I, J) = \sum_{w \in I}\dist_G(w, f(w))$.
		\begin{algorithmic}[1]
			\For{$i=1$ to $\deg_G(v)$\label{algo-line:find-target-assignment-IL2-s}}
			\While{$I_{L_i} \neq \emptyset$ and $J_{L_i} \neq \emptyset$}
			\State Let $x \in I_{L_i}$ be such that $\dist_G(x, v) = \max_{x^\prime \in I_{L_i}}\dist_G(x^\prime, v)$. \Comment{$x$ is the farthest vertex from $v$ in $I_{L_i}$ that has not yet been assigned}
			\State Let $y \in J_{L_i}$ be such that $\dist_G(y, v) = \max_{y^\prime \in J_{L_i}}\dist_G(y^\prime, v)$. \Comment{$y$ is the farthest vertex from $v$ in $J_{L_i}$ that has not yet been assigned}
			\State $f(x) \leftarrow y$; $I_{L_i} \leftarrow I_{L_i} \setminus \{x\}$; $J_{L_i} \leftarrow J_{L_i} \setminus \{y\}$. \Comment{assign $y = f(x)$ and remove them from the independent sets}
			\EndWhile 
			\EndFor \label{algo-line:find-target-assignment-IL2-e}
			\While{$\bigcup_{i = 1}^{\deg_G(v)}I_{L_i} \neq \emptyset$ and $\bigcup_{i = 1}^{\deg_G(v)}J_{L_i} \neq \emptyset$\label{algo-line:find-target-assignment-IL1-s}} \Comment{From this point, for any leg $L$, either $I_L = \emptyset$ or $J_L = \emptyset$.}
			\State Take a leg $L_i$ such that there exists $x \in I_{L_i}$ satisfying $\dist_G(x, v) = \min_{x^\prime \in \bigcup_{i = 1}^{\deg_G(v)}I_{L_i}}\dist_G(x^\prime, v)$. \Comment{$x$ is a closest vertex from $v$ in $I_{L_i}$ that has not yet been assigned}
			\State Take a leg $L_j$ such that there exists $y \in J_{L_j}$ satisfying $\dist_G(y, v) = \max_{y^\prime \in \bigcup_{i = 1}^{\deg_G(v)}J_{L_i}}\dist_G(y^\prime, v)$. \Comment{$y$ is a closest vertex from $v$ in $J_{L_i}$ that has not yet been assigned}
			\State $f(x) \leftarrow y$; $I_{L_i} \leftarrow I_{L_i} \setminus \{x\}$; $J_{L_j} \leftarrow J_{L_j} \setminus \{y\}$.
			\EndWhile \label{algo-line:find-target-assignment-IL1-e}
			\State \Return $f$.
		\end{algorithmic}
	\end{algorithm}
	The next lemma says that Algorithm~\ref{algo:find-target-assignment} efficiently produces our desired target assignment.
	\begin{lemma}
		\label{lem:algo-1-construct-shortest-assignment}
		Let $(G, I, J)$ be an instance of {\sc Shortest Sliding Token} where $I, J$ are independent sets of a spider $G$ with body $v$.
		Let $f: I \to J$ be a target assignment produced from Algorithm~\ref{algo:find-target-assignment}.
		Then, 
		\begin{itemize}
			\item[(i)] Algorithm~\ref{algo:find-target-assignment} constructs $f$ in $O(\msize{I})$ time; and
			\item[(ii)] for an arbitrary target assignment $g: I \to J$, $\sum_{w \in I}\dist_G(w, g(w)) \geq \sum_{w \in I}\dist_G(w, f(w))$.
			In other words, $f$ satisfies $\Mstar(G, I, J) = \sum_{w \in I}\dist_G(w, f(w))$.
		\end{itemize}
	\end{lemma}
	
	Before proving Lemma~\ref{lem:algo-1-construct-shortest-assignment}, we prove the following useful lemma.
	\begin{lemma}
		\label{lem:switch-assignment-pair-reduces-sum-of-dist}
		Let $(G, I, J)$ be an instance of {\sc Shortest Sliding Token} where $I, J$ are independent sets of a spider $G$ with body $v$.
		Let $f$ be a target assignment produced from Algorithm~\ref{algo:find-target-assignment}.
		Let $I = \{w_1, w_2, \dotsc, w_{\msize{I}}\}$ be such that $w_1 < w_2 < \dots <  w_{\msize{I}}$.
		Let $i, j, p \in \{1, 2, \dotsc, \msize{I}\}$ be the indices such that $w_i < \min_{<}\{w_j, w_p\}$, i.e., $w_i$ is assigned before $w_j$ and $w_p$.
		Then, $\dist_G(w_i, f(w_p)) + \dist_G(w_j, f(w_i)) \geq \dist_G(w_i, f(w_i)) + \dist_G(w_j, f(w_p))$.
	\end{lemma}
	
	\begin{proof}
		If $w_i = f(w_i)$, the desired inequality becomes the famous triangle inequality. 
		Thus, we can assume without loss of generality that $w_i \neq f(w_i)$.
		
		Based on the possible relative positions of $w_i$, $w_j$, $f(w_i)$, and $f(w_p)$, we consider the following cases.
		
		\begin{itemize}
			\item \textbf{Case 1:} $w_i$ \textbf{and} $w_j$ \textbf{are in} $I_L$.
			\begin{itemize}
				\item \textbf{Case 1.1:} $f(w_i)$ \textbf{and} $f(w_p)$ \textbf{are in} $J_L$.
				From Algorithm~\ref{algo:find-target-assignment}, we always have $\dist_G(v, f(w_i)) > \dist_G(v, f(w_p))$.
				Since $w_i \in I_L$ is assigned before $w_j \in I_L$, and $f(w_i) \in J_L$, it follows that $\dist_G(v, w_i) > \dist_G(v, w_j)$.
				(See Figure~\ref{fig:wi-wj-fi-fp-in-L}.) 	
				We note that the case $w_i = f(w_p)$ can be seen as a special case of Cases (a) or (b) in Figure~\ref{fig:wi-wj-fi-fp-in-L}.
				Similarly, the case $w_j = f(w_p)$ can be seen as a special case of Cases (b), (c), (d), or (f) in Figure~\ref{fig:wi-wj-fi-fp-in-L}; and the case $w_j = f(w_i)$ can be seen as a special case of Cases (d) or (e).
				Similar arguments hold for the next cases.
				\begin{figure}[!ht]
					\centering
					\includegraphics[scale=0.7]{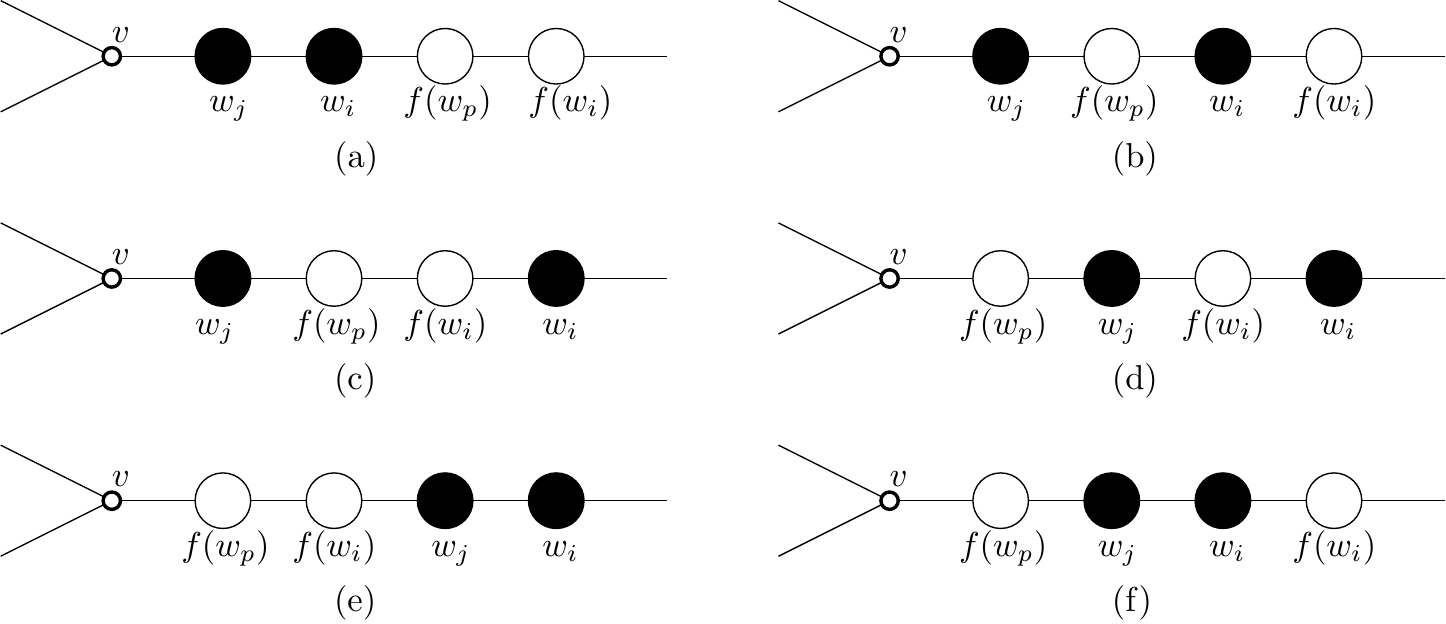}
					\caption{Possible relative positions of $w_i$, $w_j$, $f(w_i)$, and $f(w_p)$ when $w_i, w_j \in I_L$, and $f(w_i), f(w_p) \in J_L$.}
					\label{fig:wi-wj-fi-fp-in-L}
				\end{figure}
				
				\item \textbf{Case 1.2:} $f(w_i)$ \textbf{and} $f(w_p)$ \textbf{are in} $J_{L^\prime}$, $L^\prime \neq L$. (See Figure~\ref{fig:wi-wj-in-L-fi-fp-in-Lp}.)
				From Algorithm~\ref{algo:find-target-assignment}, we always have $\dist_G(v, f(w_i)) > \dist_G(v, f(w_p))$.
				Note that if $p = j$, Case (a) of Figure~\ref{fig:wi-wj-in-L-fi-fp-in-Lp} does not happen; otherwise, $w_j$ must be assigned before $w_i$.
				\begin{figure}[!ht] 
					\centering
					\includegraphics[scale=0.7]{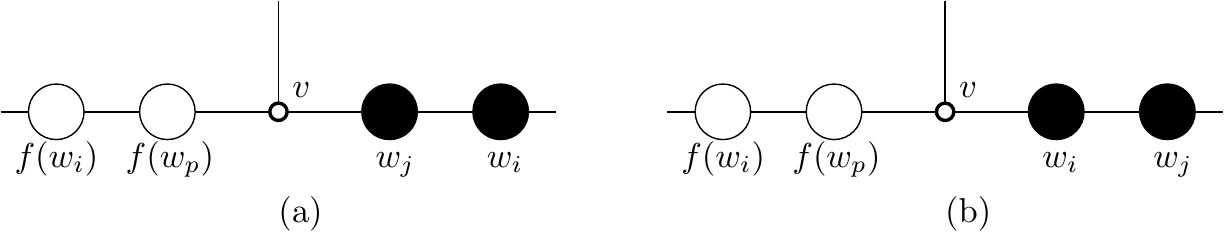}
					\caption{Possible relative positions of $w_i$, $w_j$, $f(w_i)$, and $f(w_p)$ when $w_i, w_j \in I_L$, and $f(w_i), f(w_p) \in J_{L^\prime}$, $L^\prime \neq L$.}
					\label{fig:wi-wj-in-L-fi-fp-in-Lp}
				\end{figure}
				
				\item \textbf{Case 1.3:} $f(w_i)$ \textbf{is in} $J_L$, $f(w_p)$ \textbf{is in} $J_{L^\prime}$, $L^\prime \neq L$.
				Since $w_i \in I_L$, $f(w_i) \in J_L$, and $w_i$ is assigned before $w_j$, it follows that $\dist_G(v, w_i) > \dist_G(v, w_j)$.
				(See Figure~\ref{fig:wi-wj-fi-in-L-fp-in-Lp}.) 		
				\begin{figure}[!ht]
					\centering
					\includegraphics[scale=0.7]{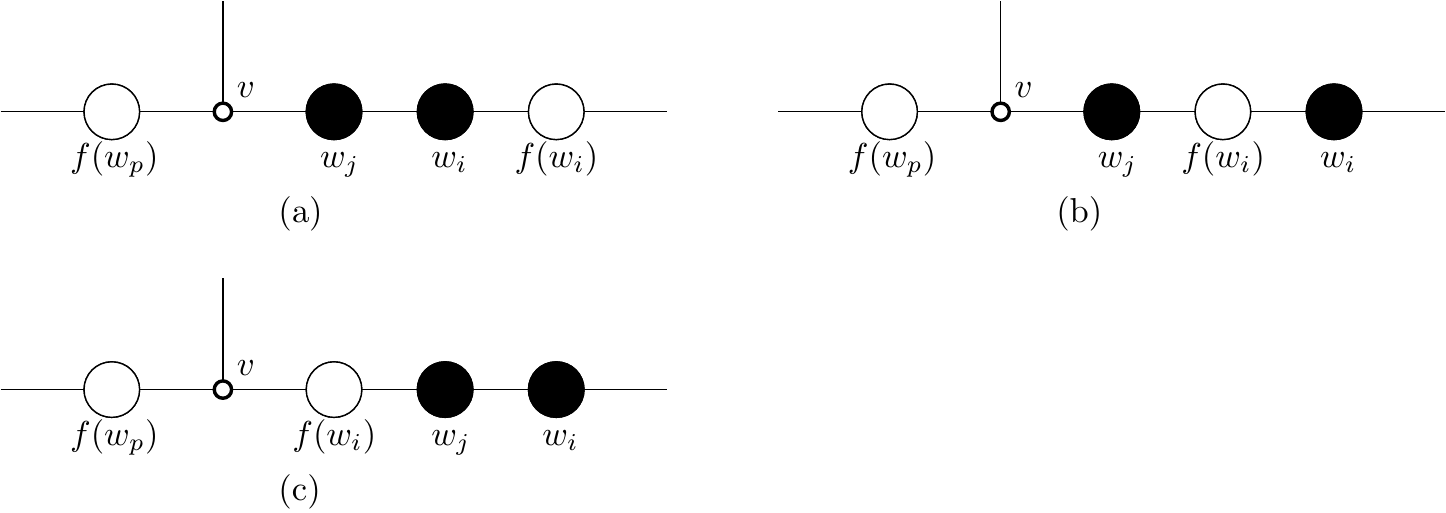}
					\caption{Possible relative positions of $w_i$, $w_j$, $f(w_i)$, and $f(w_p)$ when $w_i, w_j \in I_L$, $f(w_i) \in J_L$, and $f(w_p) \in J_{L^\prime}$, $L^\prime \neq L$.}
					\label{fig:wi-wj-fi-in-L-fp-in-Lp}
				\end{figure}
				
				\item \textbf{Case 1.4:} $f(w_p)$ \textbf{is in} $J_L$, $f(w_i)$ \textbf{is in} $J_{L^\prime}$, $L^\prime \neq L$.
				Since $w_i \in I_L$, $f(w_i) \notin J_L$, and $w_i$ is assigned before $w_j$, it follows that $\dist_G(v, w_i) < \dist_G(v, w_j)$ and $f(w_j) \notin J_L$, which implies that $p \neq j$. (See Figure~\ref{fig:wi-wj-fp-in-L-fi-in-Lp}.) 	
				\begin{figure}[!ht]
					\centering
					\includegraphics[scale=0.7]{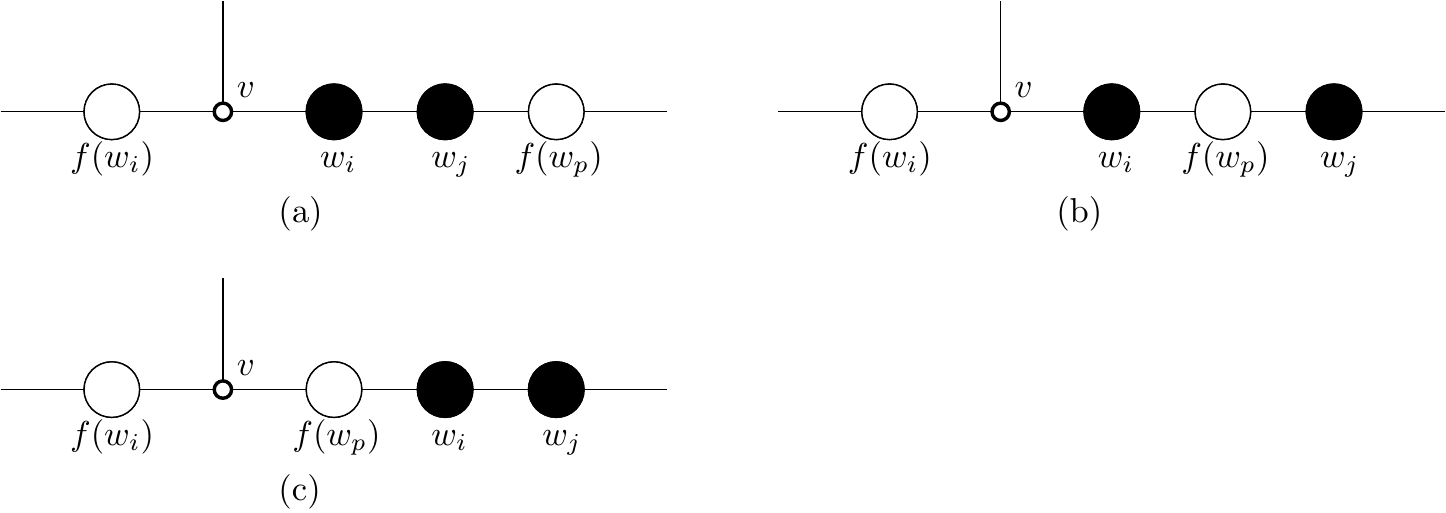}
					\caption{Possible relative positions of $w_i$, $w_j$, $f(w_i)$, and $f(w_p)$ when $w_i, w_j \in I_L$, $f(w_p) \in J_L$, and $f(w_i) \in J_{L^\prime}$, $L^\prime \neq L$.}
					\label{fig:wi-wj-fp-in-L-fi-in-Lp}
				\end{figure}
				
				\item \textbf{Case 1.5:} $f(w_i)$ is in $J_{L^\prime}$, $L^\prime \neq L$, $f(w_p)$ \textbf{is in} $J_{L^{\prime\prime}}$, $L^{\prime\prime} \notin \{L, L^\prime\}$.
				(See Figure~\ref{fig:wi-wj-in-L-fi-in-Lp-fp-in-Lpp}.)
				Note that if $p = j$, Case (b) of Figure~\ref{fig:wi-wj-in-L-fi-in-Lp-fp-in-Lpp} does not happen; otherwise, $w_j$ must be assigned before $w_i$.
				\begin{figure}[!ht]
					\centering
					\includegraphics[scale=0.7]{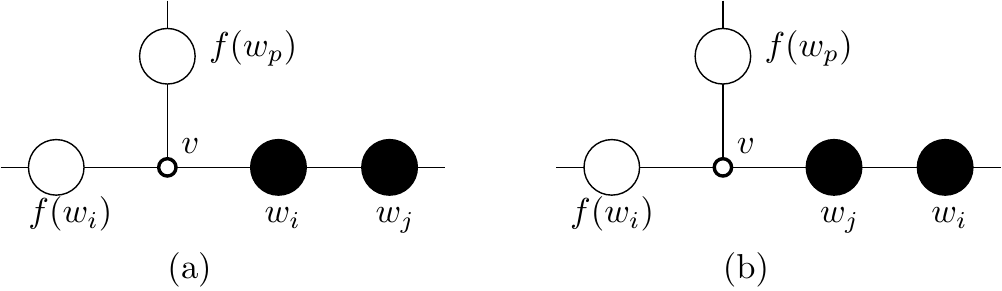}
					\caption{Possible relative positions of $w_i$, $w_j$, $f(w_i)$, and $f(w_p)$ when $w_i, w_j \in I_L$, $f(w_i) \in J_{L^\prime}$, $L^\prime \neq L$, and $f(w_p) \in J_{L^{\prime\prime}}$, $L^{\prime\prime} \notin \{L, L^\prime\}$.}
					\label{fig:wi-wj-in-L-fi-in-Lp-fp-in-Lpp}
				\end{figure}
			\end{itemize}
			
			\item \textbf{Case 2:} $w_i \in I_L$ \textbf{and} $w_j \in I_{L^\prime}$, $L \neq L^\prime$.
			\begin{itemize}
				\item \textbf{Case 2.1:} $f(w_i)$ \textbf{and} $f(w_p)$ \textbf{are in} $J_L$. 
				From Algorithm~\ref{algo:find-target-assignment}, we always have $\dist_G(v, f(w_i)) > \dist_G(v, f(w_p))$. 
				(See Figure~\ref{fig:wi-fi-fp-in-L-wj-in-Lp}.) 
				\begin{figure}[!ht]
					\centering
					\includegraphics[scale=0.7]{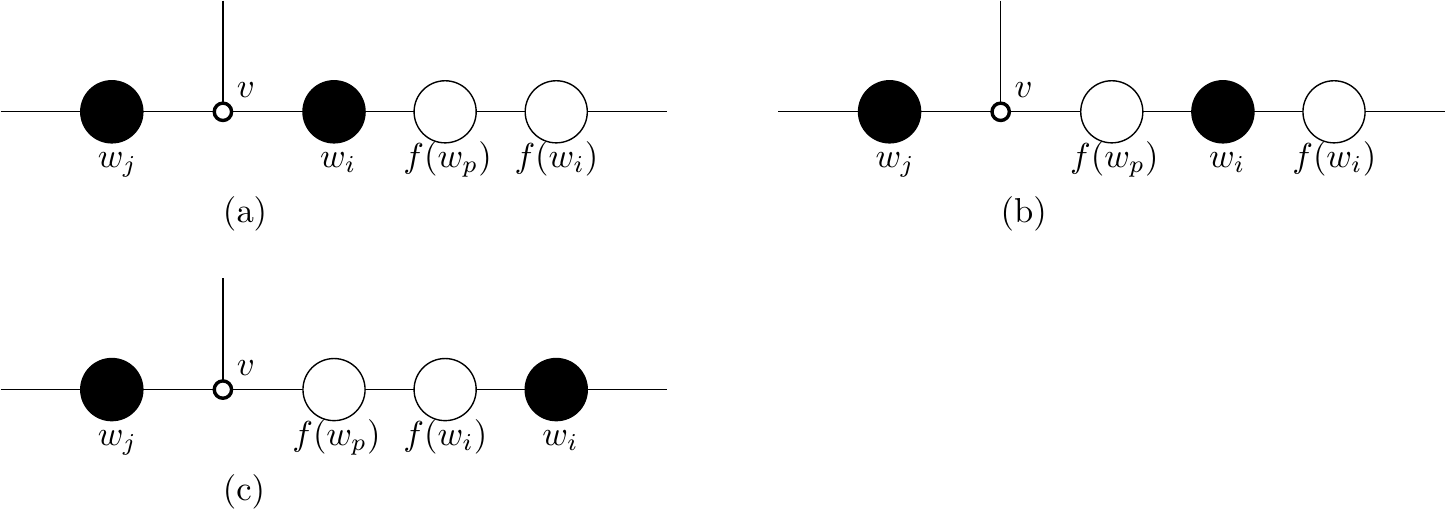}
					\caption{Possible relative positions of $w_i$, $w_j$, $f(w_i)$, and $f(w_p)$ when $w_i \in I_L$, $f(w_i), f(w_p) \in J_L$, and $w_j \in I_{L^\prime}$, $L^\prime \neq L$.}
					\label{fig:wi-fi-fp-in-L-wj-in-Lp}
				\end{figure}
				
				\item \textbf{Case 2.2:} $f(w_i)$ \textbf{and} $f(w_p)$ \textbf{are in} $J_{L^\prime}$, $L^\prime \neq L$. 
				From Algorithm~\ref{algo:find-target-assignment}, we always have $\dist_G(v, f(w_i)) > \dist_G(v, f(w_p))$. 
				Since $w_i \in I_L$ and $f(w_i) \in J_{L^\prime}$, $L^\prime \neq L$, Algorithm~\ref{algo:find-target-assignment} will assign any vertex in $I_{L^\prime}$ to some vertex in $J_{L^\prime}$, which implies $f(w_j) \in J_{L^\prime}$. 
				However, since $w_j \in I_{L^\prime}$ and $f(w_j) \in J_{L^\prime}$, Algorithm~\ref{algo:find-target-assignment} must assign $w_j$ before $w_i$, a contradiction. Thus, this case cannot happen.
				
				\item \textbf{Case 2.3:} $f(w_i)$ \textbf{and} $f(w_p)$ \textbf{are in} $J_{L^{\prime\prime}}$, $L^{\prime\prime} \notin \{L, L^\prime\}$.  
				From Algorithm~\ref{algo:find-target-assignment}, we always have $\dist_G(v, f(w_i)) > \dist_G(v, f(w_p))$. 
				(See Figure~\ref{fig:wi-in-L-wj-in-Lp-fi-fp-in-Lpp}.) 
				\begin{figure}[!ht]
					\centering
					\includegraphics[scale=0.7]{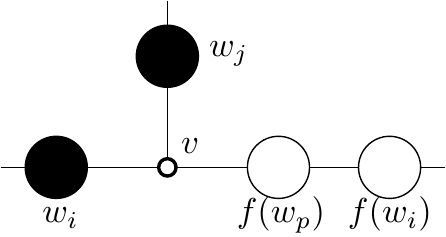}
					\caption{Possible relative positions of $w_i$, $w_j$, $f(w_i)$, and $f(w_p)$ when $w_i \in I_L$, $w_j \in I_{L^\prime}$, $L^\prime \neq L$, and $f(w_i), f(w_p) \in J_{L^\prime}$, $L^{\prime\prime} \notin \{L, L^\prime\}$.}
					\label{fig:wi-in-L-wj-in-Lp-fi-fp-in-Lpp}
				\end{figure}
				
				\item \textbf{Case 2.4:} $f(w_i)$ \textbf{is in} $J_L$, $f(w_p)$ \textbf{is in} $J_{L^\prime}$, $L^\prime \neq L$. 
				(See Figure~\ref{fig:wi-fi-in-L-wj-fp-in-Lp}.) 
				\begin{figure}[!ht]
					\centering
					\includegraphics[scale=0.7]{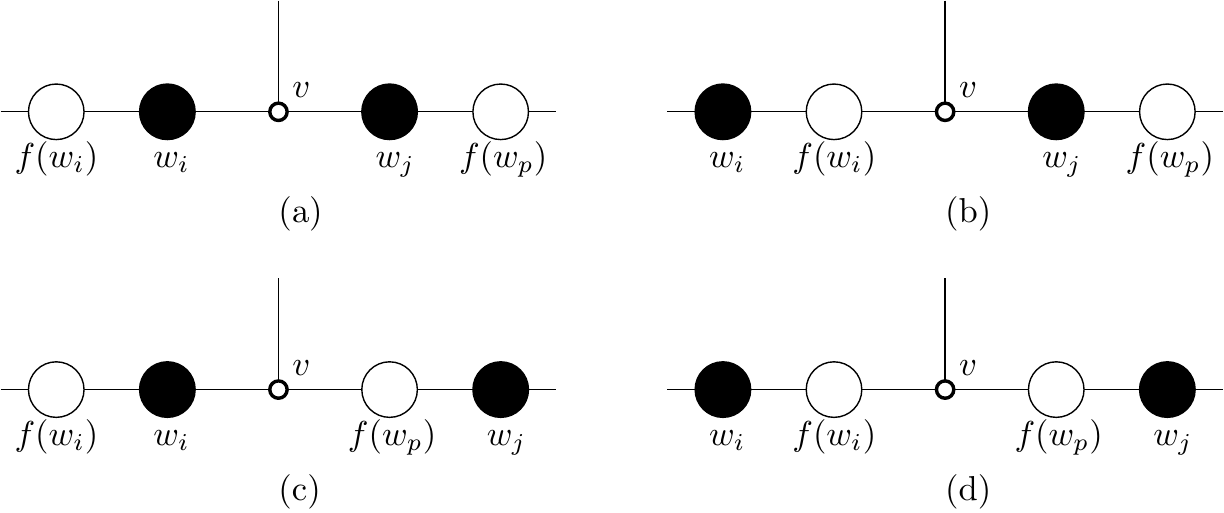}
					\caption{Possible relative positions of $w_i$, $w_j$, $f(w_i)$, and $f(w_p)$ when $w_i \in I_L$, $f(w_i) \in J_L$, $w_j \in I_{L^\prime}$, and $f(w_p) \in J_{L^\prime}$, $L^\prime \neq L$.}
					\label{fig:wi-fi-in-L-wj-fp-in-Lp}
				\end{figure}
				
				\item \textbf{Case 2.5:} $f(w_i)$ \textbf{is in} $J_L$, $f(w_p)$ \textbf{is in} $J_{L^{\prime\prime}}$, $L^{\prime\prime} \notin \{L, L^\prime\}$. 
				(See Figure~\ref{fig:wi-fi-in-L-wj-in-Lp-fp-in-Lpp}.) 
				\begin{figure}[!ht]
					\centering
					\includegraphics[scale=0.7]{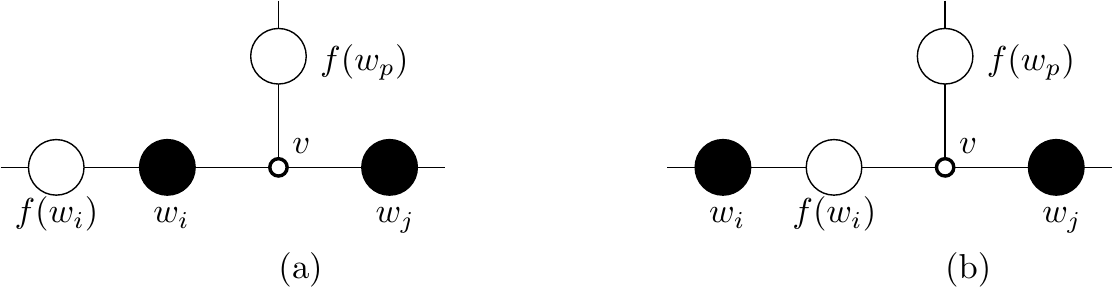}
					\caption{Possible relative positions of $w_i$, $w_j$, $f(w_i)$, and $f(w_p)$ when $w_i \in I_L$, $f(w_i) \in J_L$, $w_j \in I_{L^\prime}$, $L^\prime \neq L$, and $f(w_p) \in J_{L^{\prime\prime}}$, $L^{\prime\prime} \notin \{L, L^\prime\}$.}
					\label{fig:wi-fi-in-L-wj-in-Lp-fp-in-Lpp}
				\end{figure}
				
				\item \textbf{Case 2.6:} $f(w_i)$ \textbf{is in} $J_{L^\prime}$, $f(w_p)$ \textbf{is either in} $J_L$ \textbf{or in} $J_{L^{\prime\prime}}$, $L^{\prime\prime} \notin \{L, L^\prime\}$. 
				Since $w_i \in I_L$ and $f(w_i) \in J_{L^\prime} \neq J_L$, Algorithm~\ref{algo:find-target-assignment} assigns any $w \in I_{L^\prime}$ to some vertex in $J_{L^\prime}$, which means $f(w_j) \in J_{L^\prime}$. 
				However, this implies that $w_j$ must be assigned before $w_i$, a contradiction. Thus, this case cannot happen.
				
				\item \textbf{Case 2.7:} $f(w_i)$ \textbf{is in} $J_{L^{\prime\prime}}$, $L^{\prime\prime} \notin \{L, L^\prime\}$, $f(w_p)$ is in $J_{L^{\prime\prime\prime}}$, $L^{\prime\prime\prime} \notin \{L, L^\prime, L^{\prime\prime}\}$. 
				(See Figure~\ref{fig:wi-wj-fi-fp-in-diff-leg}.)
				\begin{figure}[!ht]
					\centering
					\includegraphics[scale=0.7]{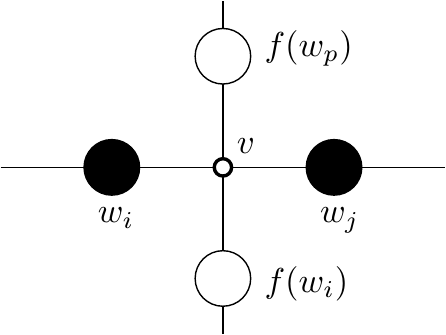}
					\caption{Possible relative positions of $w_i$, $w_j$, $f(w_i)$, and $f(w_p)$ when $w_i \in I_L$, $w_j \in I_{L^\prime}$, $f(w_i) \in J_{L^{\prime\prime}}$, and $f(w_p) \in J_{L^{\prime\prime\prime}}$.}
					\label{fig:wi-wj-fi-fp-in-diff-leg}
				\end{figure} 
			\end{itemize}
		\end{itemize}
		In all cases above, it is not hard to see that either our desired inequality holds or the case cannot happen.
		Thus, our proof is complete.
	\end{proof}
	
	We are now ready to prove Lemma~\ref{lem:algo-1-construct-shortest-assignment}.
	
	\begin{proof}[Proof of Lemma~\ref{lem:algo-1-construct-shortest-assignment}]
		Since Algorithm~\ref{algo:find-target-assignment} assigns each vertex in $I$ exactly once, (i) is trivial.
		It remains to show (ii).
		Without loss of generality, assume that $I = \{w_1, w_2, \dotsc, w_{\msize{I}}\}$ is such that $w_1 < w_2 < \dots < w_{\msize{I}}$.
		
		For an arbitrary target assignment $g$ and a target assignment $f$ produced by Algorithm~\ref{algo:find-target-assignment}, define
		$\mathsf{k}_{gf} = \msize{\{w_i \in I: g(w_i) \neq f(w_i)\}}$.
		Note that for $g \neq f$, we have $2 \leq \mathsf{k}_{gf} \leq \msize{I}$.
		We prove (ii) by induction on $\mathsf{k}_{gf}$. 
		
		\textbf{Base case:} $\mathsf{k}_{gf} = 2$. 
		It must happen that there exist $i, j$ with $i < j$, $g(w_i) = f(w_j)$, $g(w_j) = f(w_i)$, and $g(w_\ell) = f(w_\ell)$ for $w_\ell \in I \setminus \{w_i, w_j\}$.
		It follows from Lemma~\ref{lem:switch-assignment-pair-reduces-sum-of-dist} that 
		$\dist_G(w_i, f(w_j)) + \dist_G(w_j, f(w_i)) \geq \dist_G(w_i, f(w_i)) + \dist_G(w_j, f(w_j))$.
		Hence,
		$\dist_G(w_i, g(w_i)) + \dist_G(w_j, g(w_j)) \geq \dist_G(w_i, f(w_i)) + \dist_G(w_j, f(w_j))$.
		
		\textbf{Inductive step:} Given a target assignment $f$ produced from Algorithm~\ref{algo:find-target-assignment} and any target assignment $g$, suppose that for $2 \leq \mathsf{k}_{gf} \leq k-1$, 
		$\sum_{i=1}^{\msize{I}}\dist_G(w_i, g(w_i)) \geq \sum_{i=1}^{\msize{I}}\dist_G(w_i, f(w_i))$.
		We show that for every target assignment $f$ produced from Algorithm~\ref{algo:find-target-assignment} and every target assignment $g$ such that $\mathsf{k}_{gf} = k \leq \msize{I}$, the above inequality holds.
		
		Suppose to the contrary that there exist a target assignment $f$ produced from Algorithm~\ref{algo:find-target-assignment} and a target assignment $g$ such that $\mathsf{k}_{gf} = k \leq \msize{I}$ and $\sum_{i=1}^{\msize{I}}\dist_G(w_i, g(w_i)) < \sum_{i=1}^{\msize{I}}\dist_G(w_i, f(w_i))$.
		Let $i$ be the smallest index such that $g(w_i) \neq f(w_i)$.
		Let $p > i$ be such that $g(w_i) = f(w_p)$.
		Let $j > i$ be such that $g(w_j) = f(w_i)$.
		We define the assignment $g^\prime$ as follows: $g^\prime(w_i) = f(w_i)$, $g^\prime(w_j) = f(w_p)$, and for $w_\ell \in I \setminus \{w_i, w_j\}$, $g^\prime(w_\ell) = g(w_\ell)$.
		Thus, $\mathsf{k}_{g^\prime f} \leq k-1$, and by inductive hypothesis,
		$\sum_{i=1}^{\msize{I}}\dist_G(w_i, g^\prime(w_i)) \geq \sum_{i=1}^{\msize{I}}\dist_G(w_i, f(w_i))$.
		Hence, 
		$\sum_{i=1}^{\msize{I}}\dist_G(w_i, g^\prime(w_i)) \geq \sum_{i=1}^{\msize{I}}\dist_G(w_i, f(w_i)) > \sum_{i=1}^{\msize{I}}\dist_G(w_i, g(w_i))$.
		By definition of $g^\prime$, it follows that
		$\dist_G(w_i, g(w_i)) + \dist_G(w_j, g(w_j)) < \dist_G(w_i, g^\prime(w_i)) + \dist_G(w_j, g^\prime(w_j))$.
		In other words, 
		$\dist_G(w_i, f(w_p)) + \dist_G(w_j, f(w_i)) < \dist_G(w_i, f(w_i)) + \dist_G(w_j, f(w_p))$.
		However, this contradicts Lemma~\ref{lem:switch-assignment-pair-reduces-sum-of-dist}.
		Our proof is now complete.
	\end{proof}
	
	We note that Algorithm~\ref{algo:find-target-assignment} works even when the legs are labeled arbitrarily.
	However, our labeling of the legs of $G$ will be useful when we use the produced target assignment for constructing a $\sfTS$-sequence of length $\Mstar(G, I, J)$ between $I$ and $J$ in $G$.
	
	\subparagraph*{Token ordering.}
	Intuitively, we want to have a total ordering $\prec$ of vertices in $I$ such that if $x \prec y$, the token placed at $x$ should be moved before the token placed at $y$. 
	Ideally, once the token is moved to its final destination, it will never be moved again.
	From Algorithm~\ref{algo:find-target-assignment}, the following natural total ordering of vertices in $I$ can be derived: for $x, y \in I$, set $x < y$ if $x$ is assigned before $y$.
	Unfortunately, such an ordering does not always satisfy our requirement.
	However, we can use it as a basis for constructing our desired total ordering of vertices in $I$.
	
	Before showing how to construct $\prec$, we define some useful notation.
	Let $f: I \to J$ be a target assignment produced from Algorithm~\ref{algo:find-target-assignment}.
	For a leg $L$ of $G$ and a vertex $x \in I_L \cup J_L$, we say that the leg $L$ {\em contains} $x$, and $x$ is {\em inside} $L$.
	For each leg $L$ of $G$, we define $I_L^1 = \{w \in I_L: f(w) \notin J_L\}$ and $I_L^2 = \{w \in I_L: f(w) \in J_L\}$.
	Roughly speaking, a token in $I_L^1$ (resp. $I_L^2$) must finally be moved to a target outside (resp. inside) the leg $L$.
	Given a total ordering $\triangleleft$ on vertices of $I$ and a vertex $x \in I$, we define $K(x, \triangleleft) = \Neiclosed{G}{P_{xf(x)}} \cap \{y \in I: x \triangleleft y\}$.
	Intuitively, if $y \in K(x, \triangleleft)$, then in order to move the token on $x$ to its final target $f(x)$, one should move the token on $y$ beforehand.
	In some sense, the token on $y$ is an ``obstacle'' that forbids moving the token on $x$ to its final target $f(x)$.
	If $x \in I_L$ for some leg $L$ of $G$, we define $K^1(x, \triangleleft) = K(x, \triangleleft) \cap I_L^1$ and $K^2(x, \triangleleft) = K(x, \triangleleft) \cap I_L^2$.
	As before, a token in $K^1(x, \triangleleft)$ (resp. $K^2(x, \triangleleft)$) must finally be moved to a target outside (resp. inside) the leg $L$ containing $x$.
	By definition, it is not hard to see that $I_L^1$ and $I_L^2$ (resp. $K^1(x, \triangleleft)$ and $K^2(x, \triangleleft)$) form a partition of $I_L$ (resp. $K(x, \triangleleft)$).
	
	Ideally, in our desired total ordering $\prec$, for any $w \in I$, we must have $K(w, \prec) = \emptyset$.
	This enables us to move tokens in a way that any token placed at $w \in I$ is moved directly to its final target $f(w)$ through the (unique) shortest path $P_{wf(w)}$ between them; and once a token is moved to its final target, it will never be moved again.
	We note that this does not always hold for the natural total ordering $<$ defined from Algorithm~\ref{algo:find-target-assignment} above.
	Therefore, a natural approach is to construct $\prec$ from $<$ by looking at all $w \in I$ with $K(w, <) \neq \emptyset$ and reversing the ordering of any pair of vertices that makes our desired moving strategy impossible.
	A formal description of this procedure is in Algorithm~\ref{algo:construct-token-ordering} below.
	
	To provide a better explanation of Algorithm~\ref{algo:construct-token-ordering}, we briefly introduce the cases that require changing the ordering $<$.
	Assume that $w \in I$ is such that $K(w, <) \neq \emptyset$.
	\begin{itemize}
		\item {\bf Ordering between $w$ and vertices in $K(w, <)$.} 
		For each $x \in K(w, <)$, originally $w < x$, but in the new ordering, $x \prec w$. 
		That is, to move the token on $w$, one should move any ``obstacle'' (which belongs to $K(w, <)$) beforehand;
		\item {\bf Ordering between vertices in $K^2(w, <)$.}
		If $K^2(w, <) \neq \emptyset$, the token on $w$ and any token in $K^2(w, <)$ must be moved to targets inside the leg $L$ containing $w$. 
		(If $f(w) \notin J_L$ then any ``obstacle'' between $w$ and $f(w)$ must be moved to targets outside $w$, which means $K^2(w, <)$ is empty.)
		Consequently, for $x, y$ in $K^2(w, <)$, if $x < y$, the token on $x$ should move after the token on $y$, that is, we should define $x \succ y$.
		\item {\bf Ordering of vertices between $K^1(w, <)$ and $K^2(w, <)$.} If both $K^1(w, <)$ and $K^2(w, <)$ are non-empty, then it is better (but not strictly required) if we move the tokens in $K^1(w, <)$ before moving any token in $K^2(w, <)$.
		Originally, vertices in $K^1(w, <)$ (whose targets is outside $L$) is assigned after those in $K^2(w, <)$ (whose targets is inside $L$) in Algorithm~\ref{algo:find-target-assignment}.
		Intuitively, this is because tokens in $K^1(w, <)$ is ``closer'' to the body vertex $v$ than those in $K^2(w, <)$, and moving tokens in $K^1(w, <)$ creates ``empty space'' in $L$ for moving tokens in $K^2(w, <)$ later.
		
		Note that when changing the ordering of vertices between $K^1(w, <)$ and $K^2(w, <)$, we also affect the ordering between vertices in $I^1_L \supseteq K^1(w, <)$. 
		However, the ordering of vertices in $I^1_L$ should remain unchanged, since Algorithm~\ref{algo:find-target-assignment} always assign vertices in $I^1_L$ whose distance is closest to the body vertex $v$ first.
		Thus, for each $x \in I^1_L \setminus K^1(w, <)$ and $y \in K^1(w, <) \cup K^2(w, <) \cup \{w\}$, we need to set $x \prec y$.
	\end{itemize}
	
	\begin{algorithm}
		\caption{Construct a total ordering $\prec$ of vertices in $I$.}
		\label{algo:construct-token-ordering}
		\textbf{Input:} The natural ordering $<$ on vertices of $I$ derived from Algorithm~\ref{algo:find-target-assignment}.\\
		\textbf{Output:} A total ordering $\prec$ of vertices in $I$.
		\begin{algorithmic}[1]
			\While{there exists $w$ such that $K(w, <) \neq \emptyset$}
			\State Let $w$ be the smallest element of $I$ with respect to $<$ such that $K(w, <) \neq \emptyset$. \label{algo-line:construct-token-ordering-s}
			\State Let $L$ be the leg of $G$ such that $w \in I_L$. 
			\For{$x \in K(w, <)$}
			\State Set $x \prec w$. 
			\EndFor
			\If{$\msize{K^2(w, <)} \geq 2$}
			\State For $x, y \in K^2(w, <)$, if $x < y$, then set $x \succ y$.
			\EndIf
			\If{$\min\{\msize{K^1(w, <)}, \msize{K^2(w, <)}\} \geq 1$}
			\State For $x \in K^1(w, <)$ and $y \in K^2(w, <)$, set $x \prec y$.		
			\EndIf
			\If{$\min\{\msize{K^1(w, <)}, \msize{I_L^1 \setminus K^1(w, <)}\} \geq 1$}
			\State For $x \in I_L^1 \setminus K^1(w, <)$ and $y \in K^1(w, <) \cup K^2(w, <) \cup \{w\}$, set $x \prec y$.
			\EndIf
			\State For $x, y \in I$ whose ordering has not been defined, if $x < y$ then set $x \prec y$.
			\State Re-define $<$ to use in the next iteration by setting $x < y$ if $x \prec y$ for every $x, y \in I$. \label{algo-line:construct-token-ordering-e}
			\EndWhile
			\State \Return The total ordering $\prec$ of vertices in $I$.
		\end{algorithmic}
	\end{algorithm}
	
	The next lemma (Lemma~\ref{lem:token-ordering}) says that Algorithm~\ref{algo:construct-token-ordering} correctly produces a total ordering $\prec$ on vertices of $I$ such that $K(w, \prec) = \emptyset$ for every $w \in I$.
	Intuitively, Lemma~\ref{lem:token-ordering}(i) and (ii) say that if $w_i \in I_L$ is the ``chosen'' vertex in line~\ref{algo-line:construct-token-ordering-s} of Algorithm~\ref{algo:construct-token-ordering} for some leg $L$ of $G$, then only a subset $K(w_i, <) \cup I_L^1 \cup \{w_i\}$ of $I_L$ contains ``candidates'' for ``re-ordering''.
	That is, the process of changing the ordering of tokens in each iteration of Algorithm~\ref{algo:construct-token-ordering} will not affect the ordering between tokens inside and outside $L$.
	Lemma~\ref{lem:token-ordering}(iii) guarantees that after ``re-ordering'', $w_i$ will never be chosen again\footnote{$K(w_i, \prec) = \emptyset$ always holds, since none of the members of $K(w_i, <)$ will ever be larger than $w_i$ in the new orderings $\prec$ produced in the next iterations.}, and the next iteration of the main \textbf{while} loop can be initiated. 
	As Algorithm~\ref{algo:construct-token-ordering} can ``choose'' at most $\msize{I}$ vertices, and each iteration involving the ``re-ordering'' of at most $O(\msize{I})$ vertices, it will finally stop and produce the desired ordering in $O(\msize{I}^2)$ time.
	\begin{lemma}
		\label{lem:token-ordering}
		Let $(G, I, J)$ be an instance of {\sc Shortest Sliding Token} for spiders, where the body $v$ of $G$ satisfies $\max\{\msize{I \cap \Nei{G}{v}}, \msize{J \cap \Nei{G}{v}}\} = 0$.
		Let $f: I \to J$ be a target assignment produced from Algorithm~\ref{algo:find-target-assignment}, and $<$ be the corresponding natural total ordering on vertices of $I$.
		Assume that $I = \{w_1, w_2, \dotsc, w_{\msize{I}}\}$ is such that $w_1 < w_2 < \dots < w_{\msize{I}}$.
		Let $w_i$ be the smallest element in $I$ (with respect to the ordering $<$) such that $K(w_i, <) \neq \emptyset$, and $L$ be the leg of $G$ such that $w_i \in I_L$.
		Then,
		
		\begin{itemize}
			\item [(i)] $K(w_i, <) \subseteq I_L$.
			Additionally, $w_i \in I_L^2$.
			
			\item[(ii)] Let $\prec$ be the total ordering of vertices in $I$ defined as in lines~\ref{algo-line:construct-token-ordering-s}--\ref{algo-line:construct-token-ordering-e} of Algorithm~\ref{algo:construct-token-ordering}, where the corresponding vertex $w$ is replaced by $w_i$. 
			Then, 
			\begin{itemize}
				\item[(ii-1)] If $x \in K(w_i, <)$, then $x > w_i$ and $x \prec w_i$.
				\item[(ii-2)] If $x, y \in K^1(w_i, <)$, then $x < y$ if and only if $x \prec y$.
				\item[(ii-3)] If $x, y \in K^2(w_i, <)$, then $x < y$ if and only if $x \succ y$.
				\item[(ii-4)] If $x \in K^1(w_i, <)$ and $y \in K^2(w_i, <)$, then $x > y$ and $x \prec y$.
				\item[(ii-5)] If $x \in I_L^1 \setminus K^1(w_i, <)$ and $y \in K^1(w_i, <)$, then $w_i < x < y$ and $x \prec y \prec w_i$.
				\item[(ii-6)] If $x \in K(w_i, <) \cup I_L^1 \cup \{w_i\}$ and $y \in I \setminus (K(w_i, <) \cup I_L^1 \cup \{w_i\})$, then $x < y$ if and only if $x \prec y$.
				\item[(ii-7)] If $x, y \in I \setminus (K(w_i, <) \cup I_L^1 \cup \{w_i\})$, then $x < y$ if and only if $x \prec y$.
			\end{itemize}
			
			\item [(iii)] Let $\prec$ be the total ordering of vertices in $I$ described in (ii).
			Then, $K(w_i, \prec) = \emptyset$.
			Moreover, if $w_j$ is the smallest element in $I$ (with respect to the ordering $\prec$) such that $K(w_j, \prec) \neq \emptyset$, then $K(w_j, \prec) = K(w_j, <)$.
		\end{itemize}
	\end{lemma}
	
	\begin{proof}
		First of all, note that $\max\{\msize{I \cap \Nei{G}{v}}, \msize{J \cap \Nei{G}{v}}\} = 0$ is equivalent to saying that both $I \cap \Nei{G}{v}$ and $J \cap \Nei{G}{v}$ are empty.
		\begin{itemize}
			\item[(i)]  We first show that $K(w_i, <) \subseteq I_L$.
			If $f(w_i) \in J_L$, then since $I \cap \Nei{G}{v} = \emptyset$, it follows that $K(w_i, <) \subseteq I \cap V(P_{w_if(w_i)}) \subseteq I_L$.
			Now, we consider the case $f(w_i) \in J_{L^\prime} \neq J_L$ for some leg $L^\prime$ of $G$.
			Let $x \in \Neiclosed{G}{P_{w_if(w_i)}} \cap I_{L^\prime}$.
			We claim that $x < w_i$, which then implies $x \notin K(w_i, <)$ and therefore $K(w_i, <) \subseteq I_L$.
			Since $w_i \in I_L$ and $f(w_i) \notin J_L$, it follows that for any $x \in I_{L^\prime}$, $f(x) \in J_{L^\prime}$, and hence by Algorithm~\ref{algo:find-target-assignment}, $x < w_i$.
			
			Now, we show that $f(w_i) \in J_L$, which by definition means $w_i \in I_L^2$.
			Suppose to the contrary that $f(w_i) \notin J_L$.
			For a vertex $w \in I_L \cap \Neiclosed{G}{P_{w_if(w_i)}}$, we must have $\dist_G(w, v) < \dist_G(w_i, v)$, which means that $w < w_i$.
			Thus, $K(w_i, <) = \emptyset$, which contradicts the definition of $w_i$.
			
			\item[(ii)] We prove (ii-4) and (ii-5). Other statements are followed immediately from Algorithm~\ref{algo:construct-token-ordering}.
			\begin{itemize}
				\item[(ii-4)] Let $x \in K^1(w_i, <)$ and $y \in K^2(w_i, <)$.
				Clearly, by Algorithm~\ref{algo:construct-token-ordering}, $x \prec y$.
				It remains to show that $x > y$ holds.
				To see this, note that $K^1(w_i, <) \subseteq I_L^1$ and $K^2(w_i, <) \subseteq I_L^2$, and Algorithm~\ref{algo:find-target-assignment} always assigns vertices in $I_L^2$ (lines~\ref{algo-line:find-target-assignment-IL2-s}--\ref{algo-line:find-target-assignment-IL2-e}) before those in $I_L^1$ (lines~\ref{algo-line:find-target-assignment-IL1-s}--\ref{algo-line:find-target-assignment-IL1-e}).
				
				\item[(ii-5)] Let $x \in I_L^1 \setminus K^1(w_i, <)$ and $y \in K^1(w_i, <)$.
				From Algorithm~\ref{algo:construct-token-ordering}, it suffices to show $w_i < x < y$.
				For every $x \in I_L^1 \setminus K^1(w_i, <)$, since $w_i \in I_L^2$, using a similar argument as in (ii-4), we have $w_i < x$.
				It remains to show that for every $x \in I_L^1 \setminus K^1(w_i, <)$ and $y \in K^1(w_i, <)$, $x < y$.
				To see this, it is sufficient to show that if $y \in K^1(w_i, <)$, for any $z \in I_L^1$ with $z > y$, we have $z \in K^1(w_i, <)$.
				(Recall that since $<$ is a total ordering on $I$, either $x < y$ or $y < x$ and here we show that the later case cannot happen.)
				Indeed, since $z \in I_L^1$ and $z > y$, Algorithm~\ref{algo:find-target-assignment} implies $\dist_G(y, v) < \dist_G(z, v) < \dist_G(w_i, v)$ (note that $w_i \in I_L^2$).
				Since $y \in K^1(w_i, <) \subseteq I_L$, we have $\dist_G(f(w_i), v) \leq \dist_G(y, v) < \dist_G(w_i, v)$.
				Hence, $\dist_G(f(w_i), v) \leq \dist_G(z, v) < \dist_G(w_i, v)$, which means $z \in K^1(w_i, <)$.
			\end{itemize}
			
			\item[(iii)] It follows immediately from Algorithm~\ref{algo:construct-token-ordering} that $K(w_i, \prec) = \emptyset$.
			It remains to show that if $w_j$ is the smallest element in $I$ (with respect to the ordering $\prec$) such that $K(w_j, \prec) \neq \emptyset$, then $K(w_j, \prec) = K(w_j, <)$.
			
			Note that if $w_{i-1}$ exists, then $w_j \succ w_{i-1}$; otherwise, it contradicts the assumption that $w_i$ is the smallest member of $I$ (with respect to the ordering $<$) such that $K_{w_i} \neq \emptyset$.
			On the other hand, by (ii), $w_j \succ w_i$ if and only if $w_j > w_i$.
			Thus, for any $w \in I$, $w \succ w_j$ if and only if $w > w_j$, which implies $K(w_j, \prec) = K(w_j, <)$.
			
			It remains to consider the case when $w_j \in K(w_i, <) \cup I_L^1 \cup \{w_i\}$.
			\begin{itemize}
				\item We first show that if $w_j \in I_L^1 \cup \{w_i\}$, then $K(w_j, \prec) = \emptyset$.
				If $w_j = w_i$, we are done.
				Let consider the case $w_j \in I_L^1$.
				We claim that for every $x \in I \cap \Neiclosed{G}{P_{w_jf(w_j)}} = (I_L \cup I_{L^\prime} \cup (I \cap \{v\})) \cap \Neiclosed{G}{P_{w_jf(w_j)}}$, we have $x \prec w_j$, which means $x \notin K(w_j, \prec)$.
				Here $L^\prime \neq L$ is the leg of $G$ such that $f(w_j) \in J_{L^\prime}$.
				\begin{itemize}
					\item By Algorithm~\ref{algo:find-target-assignment}, if $v \in I$, for every $w \in \bigcup_{L}I_L^1$, we have $v < w$. 
					By (ii), $v \prec w$.
					Since $w_j \in I_L^1$ for some leg $L$, the above arguments hold for $w_j$. 
					\item If $x \in I_{L^\prime} \cap \Neiclosed{G}{P_{w_jf(w_j)}}$ and $x \neq v$, then by (ii), Algorithm~\ref{algo:find-target-assignment}, and the assumption $I \cap \Nei{G}{v} = \emptyset$, it follows that $x < w_j$ and $x \prec w_j$.
					\item If $x \in I_L \cap \Neiclosed{G}{P_{w_jf(w_j)}}$ and $x \neq v$, then $\dist_G(x, v) < \dist_G(w_j, v)$ (because $f(w_j) \notin J_L$), and therefore $x \in I_L^1$ and $x < w_j$.
					By (ii), we have $x \prec w_j$.
				\end{itemize}
				\item Since $K^1(w_i, <) \subseteq I_L^1$, it suffices to consider $w_j \in K^2(w_i, <)$.
				In this case, we claim that $K(w_j, \prec) = K(w_j, <)$.
				Let $K^2(w_i, <) = \{x_1, x_2, \dotsc, x_\msize{K^2(w_i, <)}\}$ be such that $x_1 \prec x_2 \prec \dots \prec x_\msize{K^2(w_i, <)}$.
				Since $w_i \in I_L$ and $f(w_i) \in J_L$, it follows that $x_1 > x_2 > \dots > x_\msize{K^2(w_i, <)}$ and $\dist_G(f(x_1), v) < \dots < \dist_G(f(x_{\msize{K^2(w_i, <)}}), v) < \dist_G(f(w_i), v) < \dist_G(x_1, v) < \dots < \dist_G(x_{\msize{K^2(w_i, <)}}, v) < \dist_G(w_i, v)$.
				Since for every $p \in \{2, 3, \dotsc, \msize{K^2(w_i, <)}\}$, $\dist_G(f(x_1), v) < \dist_G(f(x_p), v) < \dist_G(x_1, v) < \dist_G(x_p, v)$, it follows that
				if $K(x_1, \prec) = \emptyset$, for every $p \in \{2, 3, \dotsc, \msize{K^2(w_i, <)}\}$, we have $K(x_p, \prec) = \emptyset$.
				Since $w_j \in K^2(w_i, <)$ and $K(w_j, \prec) \neq \emptyset$, one of the $K(x_p, \prec)$ ($p \in \{1, 2, 3, \dotsc, \msize{K^2(w_i, <)}\}$) must be non-empty.
				Hence, $K(x_1, \prec) \neq \emptyset$, and therefore $w_j = x_1$.
				On the other hand, note that for every $x \in K^2(w_i, <) \setminus \{x_1\}$, we have $x \succ x_1$ and $x \notin \Neiclosed{G}{P_{x_1f(x_1)}}$.
				Thus, for every $w \in K(w_j, \prec) = K(x_1, \prec)$ (and then $w \succ w_j$), it must happen that $w > w_j$.
				Moreover, $w_j = x_1$ is the maximum element in $K^2(w_i, <)$ with respect to the ordering $<$.
				Thus, for every $w > w_j$, $w \notin K^2(w_i, <)$, and therefore $w \succ w_j$.
				Hence, $K(w_j, \prec) = K(w_j, <)$.
			\end{itemize}
		\end{itemize}
	\end{proof}
	
	Now, we are ready to prove the following lemma.
	\begin{lemma}
		\label{lem:no-tokens-in-neigh-v}
		Let $(G, I, J)$ be an instance of {\sc Shortest Sliding Token} for spiders where the body $v$ of $G$ satisfies $\max\{\msize{I \cap \Nei{G}{v}}, \msize{J \cap \Nei{G}{v}}\} = 0$.
		Assume that there exists a leg $L$ of $G$ with $\msize{I_L} \neq \msize{J_L}$.
		Then, in $O(n^2)$ time, one can construct a $\sfTS$-sequence $\calS$ between $I$ and $J$ such that $\len(\calS) = \Mstar(G, I, J)$.
	\end{lemma}
	
	\begin{proof}
		Let $f$ be a target assignment produced from Algorithm~\ref{algo:find-target-assignment} and $\prec$ be a corresponding total ordering defined in Algorithm~\ref{algo:construct-token-ordering}.
		For convenience, for $x, y \in I$, if $x < y$ and $x \prec y$, we say that Algorithm~\ref{algo:construct-token-ordering} \emph{preserves} the ordering between $x$ and $y$.
		
		Assume that $I = \{w_1, \dotsc, w_{\msize{I}}\}$ is such that $w_1 \prec \dots \prec w_{\msize{I}}$.
		Let $\calS$ be a sequence of token-slides constructed as follows: for each $w_i \in I$ ($i \in \{1, 2, \dotsc, \msize{I}\}$), slide the token $t_i$ on $w_i$ to $f(w_i)$ along the path $P_{w_if(w_i)}$ (using exactly $\dist_G(w_i, f(w_i))$ token-slides).
		Clearly, $\len(\calS) = \Mstar(G, I, J)$ (Lemma~\ref{lem:algo-1-construct-shortest-assignment}).
		
		Since Algorithm~\ref{algo:find-target-assignment} takes $O(n)$ time, Algorithm~\ref{algo:construct-token-ordering} takes $O(n^2)$ time, and $\calS$ uses $O(n)$ token-slides for each token in $I$, it follows that the construction of $\calS$ takes $O(n^2)$ time.
		
		To conclude this proof, we show that $\calS$ is actually a $\sfTS$-sequence in $G$ by induction on $i \in \{1, 2, \dotsc, \msize{I}\}$.
		
		\textbf{Base case:} $i = 1$.
		Since for every $j > 1$, $w_j \notin \Neiclosed{G}{P_{w_1f(w_1)}}$, $t_1$ clearly can be slid from $w_1$ to $f(w_1)$ along $P_{w_1f(w_1)}$.
		
		\textbf{Inductive step:} Assume that for $j \leq i - 1$ ($i \in \{2, 3, \dotsc, \msize{I}\}$), $t_j$ can be slid from $w_j$ to $f(w_j)$ along $P_{w_jf(w_j)}$.
		We show that $t_i$ can be slid from $w_i$ to $f(w_i)$ along $P_{w_if(w_i)}$.
		Suppose to the contrary that it is not.
		Note that by Algorithm~\ref{algo:construct-token-ordering}, for every $j > i$, $w_j \notin \Neiclosed{G}{P_{w_if(w_i)}}$.
		Thus, there must be some $j < i$ such that $f(w_j) \in \Neiclosed{G}{P_{w_if(w_i)}}$.
		In other words, after $t_j$ is moved from $w_j$ to $f(w_j)$, it becomes an ``obstacle'' that forbids sliding $t_i$ from $w_i$ to $f(w_i)$.
		We consider the following cases.
		
		\begin{itemize}
			\item \textbf{Case~1: $w_i \in I_L$ and $f(w_i) \in J_L$ for some leg $L$ of $G$.}
			Since $f(w_j) \in \Neiclosed{G}{P_{w_if(w_i)}}$, we have $f(w_j) \in J_L$.
			Since $t_j$ is moved before $t_i$ from $w_j$ to $f(w_j) \in J_L$, it follows that $w_j \in I_L$.
			Indeed, if $w_j \notin I_L$, the vertex $w_i$ must be assigned before $w_j$ in Algorithm~\ref{algo:find-target-assignment}, and Algorithm~\ref{algo:construct-token-ordering} preserves that ordering (see Lemma~\ref{lem:token-ordering}(ii)), which means $w_i \prec w_j$, a contradiction.
			Now, if $\dist_G(w_j, v) > \dist_G(w_i, v)$, we must have $\dist_G(f(w_j), v) > \dist_G(w_i, v)$; otherwise $t_i$ is an obstacle that forbids sliding $t_j$ to $f(w_j)$, which contradicts our inductive hypothesis.
			We note that the existence of $f(w_j)$ implies that $w_i \neq f(w_i)$.
			If $\dist_G(f(w_i), v) < \dist_G(w_i, v)$, then $f(w_j) \in \Neiclosed{G}{w_i}$, which also means $w_i \in \Neiclosed{G}{P_{w_jf(w_j)}}$.
			This contradicts the assumption $w_j \prec w_i$.
			Therefore, $\dist_G(f(w_i), v) > \dist_G(w_i, v)$. 	
			Since $f(w_j) \in \Neiclosed{G}{w_i}$ and $\dist_G(f(w_j), v) > \dist_G(w_i, v)$, we must have $\dist_G(w_i, v) < \dist_G(f(w_j), v) < \dist_G(f(w_i), v)$. 	
			On the other hand, since $w_i, w_j \in I_L$, $f(w_i), f(w_j) \in J_L$, and $\dist_G(w_j, v) > \dist_G(w_i, v)$, by Algorithm~\ref{algo:find-target-assignment}, we must have $\dist_G(f(w_j), v) > \dist_G(f(w_i), v)$, which is a contradiction. 	 	
			Using a similar argument, one can show that the case $\dist_G(w_j, v) < \dist_G(w_i, v)$ also leads to a contradiction. 
			
			\item \textbf{Case~2: $w_i \in I_L$ and $f(w_i) \in J_{L^\prime}$, where $L$ and $L^\prime$ are two distinct legs of $G$.}
			First of all, recall that the legs $L_i$ are labeled such that $\msize{I \cap V(L_i)} - \msize{J \cap V(L_i)} \leq \msize{I \cap V(L_j)} - \msize{J \cap V(L_j)}$ for $1 \leq i \leq j \leq \deg_G(v)$.
			Therefore, if $\msize{I_L} \neq \msize{J_L}$ for some leg $L$ of $G$, we have $\msize{I_{L_1}} \leq \msize{J_{L_1}}$.
			To see this, note that if $\msize{I_{L_1}} > \msize{J_{L_1}}$, there must be some $i \in \{2, 3, \dots, \deg_G(v)\}$ such that $\msize{I \cap V(L_i)} = \msize{I_{L_i}} < \msize{J_{L_i}} = \msize{J \cap V(L_i)}$ (because $\sum_{i=1}^{\deg_G(v)}I_{L_i} = \msize{I} = \msize{J} = \sum_{i=1}^{\deg_G(v)}J_{L_i}$), which means $\msize{I \cap V(L_i)} - \msize{J \cap V(L_i)} \leq -1$.
			On the other hand, one can verify that $\msize{I \cap V(L_1)} - \msize{J \cap V(L_1)} > -1$.
			This contradicts our assumption that $\msize{I \cap V(L_1)} - \msize{J \cap V(L_1)} \leq \msize{I \cap V(L_i)} - \msize{J \cap V(L_i)}$.
			Thus, $\msize{I_{L_1}} \leq \msize{J_{L_1}}$.
			Note that if $v \in J$ and $\msize{I_{L_1}} = \msize{J_{L_1}}$ then one can verify that $\msize{I \cap V(L_1)} - \msize{J \cap V(L_1)} \geq 0$.
			As before, there must be some $i \in \{2, 3, \dots, \deg_G(v)\}$ such that $\msize{I \cap V(L_i)} = \msize{I_{L_i}} < \msize{J_{L_i}} = \msize{J \cap V(L_i)}$, which means $\msize{I \cap V(L_i)} - \msize{J \cap V(L_i)} < 0 \leq \msize{I \cap V(L_1)} - \msize{J \cap V(V_1)}$, a contradiction.
			Therefore, if $v \in J$, we have $\msize{I_{L_1}} < \msize{J_{L_1}}$.
			It follows that if $v \in I$ (resp. $v \in J$), then $v \in I_{L_1}^2$ (resp. $f^{-1}(v) \in I_{L_1}^1$).
			Moreover, Algorithm~\ref{algo:find-target-assignment} and Lemma~\ref{lem:token-ordering}(ii) implies that 
			if $v \in J$, then $f^{-1}(v) = \max_{\prec}\{w: w \in \bigcup_LI_L^1\}$. 
			Since $w_j \prec w_i \in \bigcup_LI_L^1$, we have $w_j \neq f^{-1}(v)$, and hence $f(w_j) \neq v$.  
			Since $f(w_j) \in \Neiclosed{G}{P_{w_if(w_i)}}$, $I \cap \Nei{G}{v} = J \cap \Nei{G}{v} = \emptyset$, and $f(w_j) \neq v$, $f(w_j)$ must belong to either $J_L$ or $J_{L^\prime}$. 
			If $f(w_j)$ belongs to $J_{L^\prime}$, then since $f(w_j) \in \Neiclosed{G}{P_{w_if(w_i)}}$, it follows that $\dist_G(f(w_j), v) < \dist_G(f(w_i, v))$, which means that $w_j \notin J_{L^\prime}$. 
			It follows that $w_j$ is assigned after $w_i$, and since Algorithm~\ref{algo:construct-token-ordering} preserves this ordering, $w_j \succ w_i$, which is a contradiction.  
			Hence, $f(w_j)$ belongs to $J_L$. 
			Additionally, from Algorithm~\ref{algo:find-target-assignment}, since $f(w_i) \in J_{L^\prime} \neq J_L$, $w_j$ must belong to $I_L$ and $\dist_G(w_j, v) > \dist_G(w_i, v)$. 
			As $t_i$ is not an obstacle that forbid sliding $t_j$ from $w_j$ to $f(w_j)$, it follows that $f(w_j) \in \Neiclosed{G}{w_i}$, which means $w_i \in \Neiclosed{G}{P_{w_jf(w_j)}}$. 
			This contradicts our assumption that $w_j \prec w_i$.
		\end{itemize}
		
		Hence, $t_i$ can be slid from $w_i$ to $f(w_i)$ along $P_{w_if(w_i)}$.
		Our proof is now complete.
	\end{proof}
	
	In Lemma~\ref{lem:no-tokens-in-neigh-v}, we assumed that there is some leg $L$ of $G$ with $\msize{I_L} \neq \msize{J_L}$.
	In the next lemma, we consider the case $\msize{I_{L}} = \msize{J_{L}}$ for every leg $L$ of $G$ (regardless of whether $\max\{\msize{I \cap \Nei{G}{v}}, \msize{J \cap \Nei{G}{v}}\} = 0$).
	
	\begin{lemma}
		\label{lem:all-leg-have-same-num-of-IJ-tokens}
		Let $(G, I, J)$ be an instance of {\sc Shortest Sliding Token} for spiders.
		Let $v$ be the body of $G$.
		Assume that $\msize{I_{L}} = \msize{J_{L}}$ for every leg $L$ of $G$.
		Then, in $O(n^2)$ time, one can construct a $\sfTS$-sequence $\calS$ between $I$ and $J$ such that $\len(\calS) = \Mstar(G, I, J)$.
	\end{lemma}
	
	\begin{proof}
		Let $f$ be a target assignment produced from Algorithm~\ref{algo:find-target-assignment}.
		From the assumption, we have $\bigcup_LI_L^1 = \emptyset$, i.e., for every leg $L$, if $w \in I_L$, then $f(w) \in J_L$.
		Now, for a leg $L$, let $I_L = \{w_1, w_2, \dotsc, w_{\msize{I_L}}\}$ be such that $w_1 \prec w_2 \prec \dots \prec w_{\msize{I_L}}$, where $\prec$ is the ordering produced from Algorithm~\ref{algo:construct-token-ordering}.
		Let $\calS_L$ be a sequence of token-slides constructed as follows: for each $w_i \in I_L$ ($i \in \{1, 2, \dotsc, \msize{I_L}\}$), slide the token on $w_i$ to $f(w_i) \in J_L$ along the path $P_{w_if(w_i)}$ (using exactly $\dist_G(w_i, f(w_i))$ token-slides).
		From the proof of Lemma~\ref{lem:no-tokens-in-neigh-v} (see \textbf{Case~1}), $\calS_L$ is indeed a $\sfTS$-sequence in $G$ of length $\len(\calS_L) = \sum_{w \in I_L}\dist_G(w, f(w))$ that reconfigures $I_L$ to $J_L$, and $\calS_L$ can be constructed in $O(n^2)$ time.
		Moreover, since $\bigcup_LI_L^1 = \emptyset$, for two distinct legs $L, L^\prime$, the concatenation $\calS_L \oplus \calS_{L^\prime}$ is also a $\sfTS$-sequence in $G$.
		Thus, a $\sfTS$-sequence $\calS$ in $G$ between $I$ and $J$ can be constructed by taking the concatenation of all $\calS_L$. 
		Clearly, $\len(\calS) = \Mstar(G, I, J)$, and $\calS$ can be constructed in $O(n^2)$ time.
	\end{proof}
	
	Using Lemma~\ref{lem:all-leg-have-same-num-of-IJ-tokens}, from this point forward, we can assume without loss of generality that for an instance $(G, I, J)$ of {\sc Shortest Sliding Token} for spiders, there must be some leg $L$ of $G$ such that $\msize{I_L} \neq \msize{J_L}$.
	
	\subsection{When $0 < \max\{\msize{I \cap \Nei{G}{v}}, \msize{J \cap \Nei{G}{v}}\} \leq 1$}
	
	This section is devoted to showing the following lemma.
	\begin{lemma}
		\label{lem:one-token-in-neigh-v}
		Let $(G, I, J)$ be an instance of {\sc Shortest Sliding Token} for spiders where the body $v$ of $G$ satisfies $0 < \max\{\msize{I \cap \Nei{G}{v}}, \msize{J \cap \Nei{G}{v}}\} \leq 1$.
		Assume that $\msize{I_{L}} \neq \msize{J_{L}}$ for some leg $L$ of $G$.
		Let $x \in I$ (resp. $y \in J$) be such that $I \cap \Nei{G}{v} = \{x\}$ (resp. $J \cap \Nei{G}{v} = \{y\}$), provided that $I \cap \Nei{G}{v} \neq \emptyset$ (resp. $J \cap \Nei{G}{v} \neq \emptyset$).
		Then,
		\begin{itemize}
			\item[(i)] If $x$ and $y$ both exist, and $x \in I_L$ and $y \in J_L$ for some leg $L$ of $G$ with $\msize{I_L} = \msize{J_L}$, then for every $\sfTS$-sequence $\calS$ between $I$ and $J$, $D_G(\calS) \geq 2$.
			Consequently, $\Dstar(G, I, J) \geq 2$.
			Moreover, one can construct in $O(n^2)$ time a $\sfTS$-sequence between $I$ and $J$ of length $\Mstar(G, I, J) + 2$.
			\item[(ii)] Otherwise, one can construct in $O(n^2)$ time a $\sfTS$-sequence between $I$ and $J$ of length $\Mstar(G, I, J)$.
		\end{itemize}
	\end{lemma}
	
	\begin{proof}
		\begin{itemize}
			\item[(i)] By assumption, we have $\Nei{G}{v} \cap V(L) = \{x\} = \{y\}$.
			If $\Nei{G}{x} \setminus \{v\} \neq \emptyset$, then let $v^\prime$ be such that $\Nei{G}{x} \setminus \{v\} = \{v^\prime\}$. 
			(Note that since $G$ is spider, $\msize{\Nei{G}{x}} \leq 2$.)
			We claim that for any $\sfTS$-sequence $\calS$ in $G$ between $I$ and $J$, $\calS$ must make detour over either $e_1 = xv = yv$ or $e_2 = xv^\prime = yv^\prime$.
			(See Figure~\ref{fig:one-token-in-neigh-v}.)
			\begin{figure}[!ht]
				\centering
				\includegraphics[scale=0.7]{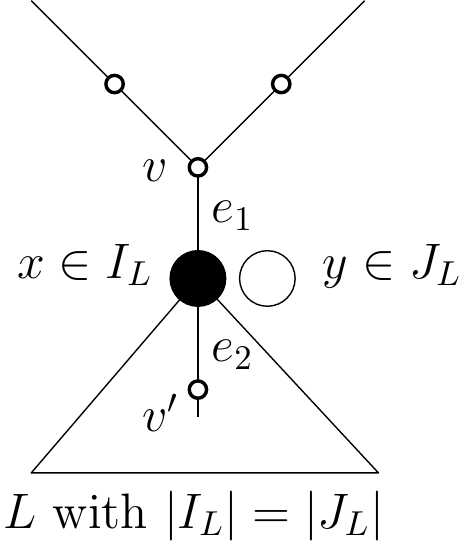}
				\caption{Illustration of Lemma~\ref{lem:one-token-in-neigh-v}(i). Tokens in $I$ (resp. $J$) are of black (resp. white) color.}
				\label{fig:one-token-in-neigh-v}
			\end{figure}
			By assumption, note that there must be some leg $K$ such that $\msize{I_K} > \msize{J_K}$.
			Let $t_w$ be the token placed at some vertex $w \in I_K$.
			Since $\msize{I_K} > \msize{J_K}$, at some point, $\calS$ must slide $t_w$ to some vertex not in $K$.
			As $G$ is a spider, at some point, $\calS$ must slide $t_w$ to $v$, which means it must slide the token $t_x$ on $x \in I_L$ to some other vertex not in $\Nei{G}{v}$ beforehand.
			There are only two possible movements: at some point, $\calS$ slides $t_x$ either to $v$ or to $v^\prime$ (and then maybe to some other vertex not in $\Nei{G}{v}$).
			\begin{itemize}
				\item \textbf{Case~1: $\calS$ slides $t_x$ to $v$ and then to some other vertex not in $\Nei{G}{v}$.}
				Let $\Ip$ be the resulting independent set at this point.
				Let $\calS_1$ and $\calS_2$ be the subsequences of $\calS$ that reconfigure $I$ to $\Ip$ and $\Ip$ to $J$, respectively.
				Clearly, $\calS_1$ slides $t_x$ from $x$ to $v$ at some point.
				Since $\msize{\Ip_L} < \msize{\Ip_L}$ ($t_x$ is already moved) and $y \in J_L$, it follows that at some point $\calS_2$ must slide some token $t_z$ on some vertex $z \ \notin \Ip_L$ to $y = x$.
				As $G$ is a spider, $\calS_1$ must slide $t_z$ to $v$ beforehand, and then moves $t_z$ from $v$ to $y = x$.
				In summary, $x \to v$ and $v \to y = x$ are members of $\calS_1$ and $\calS_2$, respectively.
				Hence, $\calS$ makes detour over $e_1 = xv$.
				
				\item \textbf{Case~2: $\calS$ slides $t_x$ to $v^\prime$ (and then maybe to some other vertex not in $\Nei{G}{v}$).}
				Let $\Ipp$ be the resulting independent set at this point.
				Let $\calS_3$ and $\calS_4$ be the subsequences of $\calS$ that reconfigure $I$ to $\Ipp$ and $\Ipp$ to $J$, respectively.
				Clearly, $\calS_3$ slides $t_x$ from $x$ to $v^\prime$ at some point.
				Since $\msize{\Ipp_L} = \msize{\Ipp_L}$ and $y \in J_L$, at some point, $\calS_4$ must slide $t_x$ from $v^\prime$ to $y = x$ (regardless of which token is finally moved to $y$).
				In summary, $x \to v^\prime$ and $v^\prime \to y = x$ are members of $\calS_3$ and $\calS_4$, respectively.
				Hence, $\calS$ makes detour over $e_2 = xv^\prime$.
			\end{itemize}
			Since for any $\sfTS$-sequence $\calS$, one of the above cases must happen, we always have $D_G(S) \geq 2$, which means $\Dstar(G, I, J) \geq 2$.
			
			Now, we describe how to construct a $\sfTS$-sequence $\calS$ whose length $\len(\calS) = \Mstar(G, I, J) + 2$.
			Let $\Ip = I \setminus \{x\} \cup \{v\}$ and $\Jp = J \setminus \{y\} \cup \{v\}$.
			Intuitively, $\Ip$ (resp. $\Jp$) is obtained from $I$ (resp. $J$) via a single token-slide that moves the token on $x \in I$ (resp. $y \in J$) to $v$. 
			Note that this can be done because $\max\{\msize{I \cap \Nei{G}{v}}, \msize{J \cap \Nei{G}{v}}\} \leq 1$.
			Recall that $x = y \in I \cap J$.
			Clearly, $\max\{\msize{\Ip \cap \Nei{G}{v}}, \msize{\Jp \cap \Nei{G}{v}}\} = 0$.
			Moreover, by Lemma~\ref{lem:algo-1-construct-shortest-assignment}, $\Mstar(G, I, J) = \Mstar(G, \Ip, \Jp)$.
			To see this, note that if $f: I \to J$ (resp. $g: \Ip \to \Jp$) is a target assignment produced from Algorithm~\ref{algo:find-target-assignment}, then $x = f(x) \in I_L \cap J_L$ and $v = g(v) \in I_{L_1} \cap J_{L_1}$.
			Let $\calSp$ be a $\sfTS$-sequence in $G$ that reconfigures $\Ip$ to $\Jp$ as described in Lemma~\ref{lem:no-tokens-in-neigh-v}.
			Let $\calS = \langle x \to v \rangle \oplus \calSp \oplus \langle v \to x \rangle$.
			Clearly, $\calS$ is a $\sfTS$-sequence in $G$ that reconfigures $I$ to $J$ of length $\len(\calS) = \Mstar(G, I, J) + 2$, and it can be constructed in $O(n^2)$ time.
			
			\item[(ii)] Let $f: I \to J$ be a target assignment produced from Algorithm~\ref{algo:find-target-assignment}, and $\prec$ be a corresponding total ordering on vertices of $I$ produced from Algorithm~\ref{algo:construct-token-ordering}.
			In each of the following cases, we describe how to construct a $\sfTS$-sequence $\calS$ in $G$ between $I$ and $J$ whose length is $\Mstar(G, I, J)$.
			\begin{itemize}
				\item \textbf{Case~1: Only $x$ exists.} 
				First of all, we consider the case $v \in J$ and $f(x) = v$.
				Let $\Ip = I \setminus \{x\} \cup \{v\}$.
				Intuitively, the independent set $\Ip$ is obtained from $I$ by sliding the token on $x$ to $v$.
				Moreover, by Lemma~\ref{lem:algo-1-construct-shortest-assignment}, $\Mstar(G, I, J) = \Mstar(G, \Ip, J) + 1$.
				To see this, note that if $g: \Ip \to J$ is a target assignment produced from Algorithm~\ref{algo:find-target-assignment}, then $g(v) = v = f(x)$.
				Since only $x$ exists, we must have $\max\{\msize{\Ip \cap \Nei{G}{v}}, \msize{J \cap \Nei{G}{v}}\} = 0$.
				Let $\calSp$ be a $\sfTS$-sequence in $G$ that reconfigures $\Ip$ to $J$ as described in Lemma~\ref{lem:no-tokens-in-neigh-v}.
				Clearly, the sequence $\calS = \langle x \to v \rangle \oplus \calSp$ is our desired $\sfTS$-sequence.
				
				Without loss of generality, we can now assume that if $v \in J$, then we have $f(x) \neq v$.
				Suppose that $x \in I_L$ for some leg $L$ of $G$.
				We consider the following cases.
				\begin{itemize}
					\item \textbf{Case~1.1: $f(x) \in J_L$.} 
					Let $I_x = \{x\} \cup \{w \in I_L: w \prec x\}$.
					Let $\calS_1$ be a sequence of token-slides constructed as follows: (a) Take the minimum element $w$ of $I_x$ (with respect to $\prec$) and slide the token on $w$ to $f(w)$; and (b) Repeat (a) with $I_x \setminus \{w\}$ instead of $I_x$.
					From the proof of Lemma~\ref{lem:no-tokens-in-neigh-v}, it follows that $\calS_1$ is indeed a $\sfTS$-sequence in $G$ of length $\sum_{w \in I_x}\dist_G(w, f(w))$ that moves the token on $x$ to $f(x)$.
					Intuitively, $\calS_1$ only moves tokens ``inside'' the leg $L$.
					Additionally, note that if $\Ip$ is the resulting independent set obtained from $I$ by performing $\calS_1$, then $\Mstar(G, I, J) = \Mstar(G, \Ip, J) + \len(\calS_1)$ and $\max\{\msize{\Ip \cap \Nei{G}{v}}, \msize{J \cap \Nei{G}{v}}\} = 0$.
					Thus, if $\calS_2$ is the $\sfTS$-sequence in $G$ that reconfigures $\Ip$ to $J$, as described in Lemma~\ref{lem:no-tokens-in-neigh-v}, then $\calS = \calS_1 \oplus \calS_2$ is our desired $\sfTS$-sequence.
					
					\item \textbf{Case~1.2: $f(x) \in J_{L^\prime}$ for some leg $L^\prime \neq L$ of $G$.}
					Let $I_x = \{x\} \cup \{w: w \in I_{L^\prime}\}$.
					From Algorithm~\ref{algo:find-target-assignment} and Lemma~\ref{lem:token-ordering}(ii), note that $x = \max_{\prec}\{w: w \in I_x\}$.
					Let $\calS_1$ be the sequence of token-slides constructed as follows: (a) Take the minimum element $w$ of $I_x$ (with respect to $\prec$) and slide the token on $w$ to $f(w)$; and (b) Repeat (a) with $I_x \setminus \{w\}$ instead of $I_x$.
					From the proof of Lemma~\ref{lem:no-tokens-in-neigh-v} and the assumption $\max\{\msize{I \cap \Nei{G}{v}}, \msize{J \cap \Nei{G}{v}}\} \leq 1$, it follows that $\calS_1$ is a $\sfTS$-sequence in $G$ of length $\sum_{w \in I_x}\dist_G(w, f(w))$ that moves the token on $x$ to $f(x)$.
					Intuitively, $\calS_1$ first moves tokens ``inside'' the leg $L^\prime$ to their final target vertices in order to ``clear the path'' for moving the token on $x$ to $f(x)$.  
					Additionally, note that if $\Ip$ is the resulting independent set obtained from $I$ by performing $\calS_1$, then $\Mstar(G, I, J) = \Mstar(G, \Ip, J) + \len(\calS_1)$ and $\max\{\msize{\Ip \cap \Nei{G}{v}}, \msize{J \cap \Nei{G}{v}}\} = 0$.
					As before, if $\calS_2$ is the $\sfTS$-sequence in $G$ that reconfigures $\Ip$ to $J$, as described in Lemma~\ref{lem:no-tokens-in-neigh-v}, then $\calS = \calS_1 \oplus \calS_2$ is our desired $\sfTS$-sequence.
				\end{itemize}
				
				\item \textbf{Case~2: Only $y$ exists.}
				First of all, we consider the case $v \in I$ and $f(v) = y$.
				Let $\Jp = J \setminus \{y\} \cup \{v\}$.
				Analogously to \textbf{Case~1}, we have $\Mstar(G, I, J) = \Mstar(G, I, \Jp) + 1$, and
				$\max\{\msize{I \cap \Nei{G}{v}}, \msize{\Jp \cap \Nei{G}{v}}\} = 0$. 
				Then, if $\calSp$ is the $\sfTS$-sequence that reconfigures $I$ to $\Jp$ as described in Lemma~\ref{lem:no-tokens-in-neigh-v} then $\calS = \calSp \oplus \langle v \to y \rangle$ is our desired $\sfTS$-sequence.
				
				Without loss of generality, we can now assume that if $v \in I$ then $f(v) \neq y$.
				Suppose that $y \in J_L$ for some leg $L$ of $G$.
				We consider the following cases.
				\begin{itemize}
					\item \textbf{Case~2.1: $f^{-1}(y) = z \in I_L$.}
					Let $I_z = \{z\} \cup \{w \in I_L: w \prec z\}$.
					From Algorithm~\ref{algo:find-target-assignment}, Lemma~\ref{lem:token-ordering}(ii), and the assumption $y \in \Nei{G}{v}$, we must have $I_z \setminus \{z\} \subseteq I_L^1$.
					Intuitively, the target of any token ``inside'' the path $P_{yz}$ must be some vertex ``outside'' $L$. 
					Then, for each $w \in I_z \setminus \{z\}$, one can use a similar idea as in \textbf{Case~1.2} for constructing a $\sfTS$-sequence $\calS_w$ that moves the token on $w \in I_L$ to $f(w) \notin J_L$. 
					For notational convention, let $\calS_z$ be the $\sfTS$-sequence of $\dist_G(z, f(z))$ token-slides that moves the token on $z \in I_L$ to $f(z) = y \in J_L$ along $P_{zf(z)}$.
					Let $\calS_1$ be a $\sfTS$-sequence of token-slides constructed as follows: (a) Take the minimum element $w$ of $I_z$ (with respect to $\prec$) and perform $\calS_w$; and (b) Repeat (a) with $I_z \setminus \{w\}$ instead of $I_z$.
					Intuitively, $\calS_1$ moves every token ``inside'' the path $P_{zy}$ (which, by Algorithm~\ref{algo:find-target-assignment}, is also in $I_L^1$) ``out of'' the leg $L$, and then moves the token on $z$ to $y$.
					Let $\Ip$ be the resulting independent set obtained from $I$ by performing $\calS_1$.
					Then, note that $\Ip \cap J = f(I_z) = \bigcup_{w \in I_z}\{f(w)\}$.
					It follows that the reverse $\sfTS$-sequence $\rev(\calS_1)$ of $\calS_1$ can be performed with the initial independent set $J$.
					Intuitively, $\rev(\calS_1)$ moves any token on $w \in f(I_z) \subseteq J$ to $f^{-1}(w) \in I$, and tokens not in $f(I_z)$ remains at their original position.
					Now, let $\Jp$ be the resulting independent set obtained from $J$ by performing $\rev(\calS_1)$.
					Clearly, $\Jp \cap I = I_z$, and therefore $\calS_1$ can be performed with the initial independent set $\Jp$.
					Intuitively, only tokens in $I_z$ are not placed at their final positions.
					Note that $\Mstar(G, I, J) = \Mstar(G, I, \Jp) + \sum_{w \in I_z}\len(\calS_w)$ and $\max\{\msize{I \cap \Nei{G}{v}}, \msize{\Jp \cap \Nei{G}{v}}\} = 0$.
					Then, if $\calS_2$ is the $\sfTS$-sequence that reconfigures $I$ to $\Jp$ as described in Lemma~\ref{lem:no-tokens-in-neigh-v} then $\calS = \calS_2 \oplus \calS_1$ is our desired $\sfTS$-sequence.
					
					\item \textbf{Case~2.2: $f^{-1}(y) = z \in I_{L^\prime}$ for some leg $L^\prime \neq L$ of $G$.}
					Let $I_z = \{z\} \cup \{w \in I_{L^\prime}: w \prec z\}$.
					From Algorithm~\ref{algo:find-target-assignment} and Lemma~\ref{lem:token-ordering}(ii), we must have $I_z \subseteq I_{L^\prime}^1$.
					Now, the construction of our desired $\sfTS$-sequence $\calS$ can be done in a similar manner as in \textbf{Case~2.1}.
				\end{itemize}
				
				\item \textbf{Case~3: Both $x$ and $y$ exist, and $f(x) \neq y$.}
				We note that in this case $v \notin I \cup J$.
				Combining the techniques in \textbf{Case~1} and \textbf{Case~2}, one can construct:
				\begin{itemize}
					\item a $\sfTS$-sequence $\calS_1$ that moves the token on $x$ to $f(x)$, and the resulting independent set $\Ip$ satisfies $\Mstar(G, I, J) = \Mstar(G, \Ip, J) + \len(\calS_1)$; and
					\item a $\sfTS$-sequence $\calS_2$ whose reverse $\rev(\calS_2)$ moves the token on $y$ to $f^{-1}(y)$, and the resulting independent set $\Jp$ after performing $\rev(\calS_2)$ satisfies $\Mstar(G, \Ip, J) = \Mstar(G, \Ip, \Jp) + \len(\calS_2)$.
				\end{itemize}
				Note that $\max\{\msize{\Ip \cap \Nei{G}{v}}, \msize{\Jp \cap \Nei{G}{v}}\} = 0$.
				Thus, if $\calS_3$ is the $\sfTS$-sequence that reconfigures $\Ip$ to $\Jp$ as described in Lemma~\ref{lem:no-tokens-in-neigh-v} then $\calS = \calS_1 \cup \calS_3 \cup \calS_2$ is our desired $\sfTS$-sequence.
				
				\item \textbf{Case~4: Both $x$ and $y$ exist, and $f(x) = y$.}
				We note that in this case $v \notin I \cup J$.
				From the assumption, it must happen that $x \in I_L$ and $y \in J_{L^\prime}$ for two distinct legs $L, L^\prime$ of $G$.
				(If $L = L^\prime$ the Algorithm~\ref{algo:find-target-assignment} implies that $\msize{I_L} = \msize{J_L}$, which contradicts our assumption.)
				Moreover, Algorithm~\ref{algo:find-target-assignment} implies that $\bigcup_LI_L^1 = \{x\}$.
				To see this, note that if there exists $z \in \bigcup_LI_L^1 \setminus \{x\}$ then we must have $1 = \dist_G(x, v) \leq \dist_G(z, v)$ and $1 = \dist_G(f(x), v) \leq \dist_G(f(z), v)$; otherwise, either $z \neq x$ or $f(z) \neq y$ belongs to $\Nei{G}{v}$, which contradicts the assumption $\max\{\msize{I \cap \Nei{G}{v}}, \msize{J \cap \Nei{G}{v}}\} \leq 1$.
				However, this contradicts Algorithm~\ref{algo:find-target-assignment}.
				Thus, we must have $\bigcup_LI_L^1 = \{x\}$.
				Let $\calS_1 = \langle x \to v, v \to y \rangle$ be the $\sfTS$-sequence of length $\dist_G(x, y) = 2$ that moves the token on $x \in I_L$ to $y \in J_{L^\prime}$, and let $\Ip$ be the resulting independent set. 
				This can be done simply because $\max\{\msize{I \cap \Nei{G}{v}}, \msize{J \cap \Nei{G}{v}}\} \leq 1$.
				Note that $\Mstar(G, I, J) = \Mstar(G, \Ip, J) + 2$, and every leg $L$ of $G$ satisfies $\msize{\Ip_L} = \msize{J_L}$ (otherwise, $\bigcup_L\Ip_L^1 \neq \emptyset$, and hence $\bigcup_LI_L^1 \neq \{x\}$, which is a contradiction).
				Then, if $\calS_2$ is the $\sfTS$-sequence that reconfigures $\Ip$ to $J$ as described in Lemma~\ref{lem:all-leg-have-same-num-of-IJ-tokens} then $\calS = \calS_1 \oplus \calS_2$ is our desired $\sfTS$-sequence.
			\end{itemize}
			We have shown how to construct a $\sfTS$-sequence $\calS$ in $G$ between $I$ and $J$ whose length is $\Mstar(G, I, J)$.
			From the above cases, it is clear that the construction of $\calS$ takes $O(n^2)$ time. 
		\end{itemize}
	\end{proof}
	
	\subsection{When $\max\{\msize{I \cap \Nei{G}{v}}, \msize{J \cap \Nei{G}{v}}\} \geq 2$}
	
	In this section, we claim that 
	\begin{lemma}
		\label{lem:two-tokens-in-neigh-v}
		Let $(G, I, J)$ be an instance of {\sc Shortest Sliding Token} for spiders where the body $v$ of $G$ satisfies $\max\{\msize{I \cap \Nei{G}{v}}, \msize{J \cap \Nei{G}{v}}\} \geq 2$.
		Assume that $\msize{I_{L}} \neq \msize{J_{L}}$ for some leg $L$ of $G$.
		Then, in $O(n^2)$ time, one can construct a $\sfTS$-sequence $\calS$ between $I$ and $J$ of shortest length.
		Moreover, the value of $D_G(\calS)$ can be explicitly calculated.
	\end{lemma}
	
	Before proving Lemma~\ref{lem:two-tokens-in-neigh-v}, we prove the following useful lemma.
	\begin{lemma}
		\label{lem:calculate-cost-in-linear-time}
		The value of $\cost(T, I, xy)$ can be calculated in $O(n)$ time for $x \in I$ and $y \in \Nei{T}{x}$, where $I$ is an independent set of a tree $T$ on $n$ vertices.
		Moreover, if $\cost(T, I, xy) < \infty$, then in $O(n)$ time, one can output a $\sfTS$-sequence $\calS(T, I, xy)$ in $T$ of length $\cost(T, I, xy)$ such that $\calS(T, I, xy)$ moves the token on $x$ to $y$.
	\end{lemma}
	
	\begin{proof}
		We modify the algorithm described in \cite[Lemma 2]{Demaine2014-TS-trees}.
		First, regard $x$ as the root of $T$.
		Then, we define $\phi(y)$ for each vertex $y \in V(T)$ from leaves of $T$ to the root $u$ as follows.
		For each leave $y$ of $T$, we set $\phi(y) = \infty$ if $y \in I$; otherwise, $\phi(y) = 1$.
		For each internal vertex $y$ of $T$ with $y \notin I$, if no children of $y$ is in $I$, we set $\phi(y) = 1$; otherwise, we set $\phi(y) = 1 + \sum_{w \in I \text{ and $w$ is a child of $y$}}\phi(w)$.
		For each internal vertex $y$ of $T$ with $y \in I$, we set $\phi(y) = \min_{\text{child $w$ of $y$}}\phi(w)$.
		Finally, we set $\cost(T, I, xy) = \phi(y)$.
		By definition, it is not hard to see that the above algorithm correctly computes $\cost(T, I, xy)$.
		For each $y \in V(T)$, the value $\phi(y)$ is computed exactly once.
		Thus, $\cost(T, I, xy)$ can be calculated in $O(n)$ time.
		
		Assume that $\cost(T, I, xy) < \infty$.
		We now show how to construct $\calS(T, I, xy)$ using the described algorithm.
		For each $z \in I \setminus \{u\}$ with $\phi(z) < \infty$, we define $c(z)$ to be a child of $z$ such that $\phi(c(z)) = \min_{\text{child $w$ of $z$}}\phi(w)$.
		For $z \in I$ with $\phi(z) = 1$, clearly $\calS(T, I, zc(z)) = \langle z \to c(z) \rangle$.
		For every $z \in I$ with $1 < \phi(z) < \infty$, set $\calS(T, I, zc(z)) = \bigoplus_{z^\prime \in I \cap \Nei{T^z_{c(z)}}{c(z)}}\calS(T, I, z^\prime c(z^\prime)) \oplus \langle z \to c(z) \rangle$.
		One can verify that the sequence $\calS(T, I, zc(z))$ of token-slides is indeed a $\sfTS$-sequence in $T$.
		The $\sfTS$-sequence $\calS(T, I, xy)$ is indeed $\bigoplus_{w \in I \cap \Nei{T^x_y}{y}}\calS(T, I, wc(w)) \oplus \langle x \to y \rangle$.
		Clearly, we can use this recursive relation to construct $\calS(T, I, xy)$ in $O(n)$ time.
	\end{proof}
	
	Next, we define a useful notation for calculating the number of detours.
	For an instance $(T, I, J)$ of {\sc Shortest Sliding Token} for trees, we define a directed \emph{auxiliary graph} $A(T, I, J)$ as follows: $V(A(T, I, J)) = V(T)$; and $E(A(T, I, J)) = \{(x, y): xy \in E(T) \text{ and } \msize{I \cap T^x_y} \leq \msize{J \cap T^x_y}\}$.
	By definition, the auxiliary graph $A(G, J, I)$ can be obtained from $A(G, I, J)$ by simply reversing the directions of its edges.
	Figure~\ref{fig:exa-auxiliary-graph} illustrates an example of the auxiliary graph $A(G, I, J)$ for an instance $(G, I, J)$ of the problem for spiders.
	\begin{figure}[!ht]
		\centering
		\includegraphics[scale=0.7]{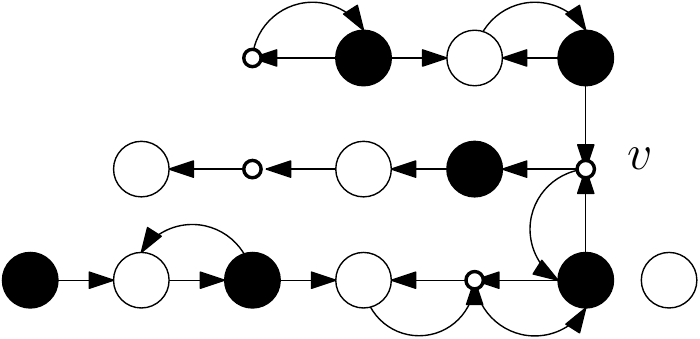}
		\caption{An example of the auxiliary graph $A(G, I, J)$ for an instance of {\sc Shortest Sliding Token} for spiders. Tokens in $I$ (resp. $J$) are of black (resp. white) color.}
		\label{fig:exa-auxiliary-graph}
	\end{figure}
	
	We are now ready to prove Lemma~\ref{lem:two-tokens-in-neigh-v}.
	\begin{proof}[Proof of Lemma~\ref{lem:two-tokens-in-neigh-v}]
		We consider the following cases.
		\begin{itemize}
			\item \textbf{Case~1: $\msize{I \cap \Nei{G}{v}} \geq 2$ and $\msize{J \cap \Nei{G}{v}} \leq 1$.} (See Figure~\ref{fig:two-tokens-in-neigh-v}.)
			
			From the assumption, note that $v \notin I$.
			Let $I \cap \Nei{G}{v} = \{w_1, w_2, \dotsc, w_k\}$ ($2 \leq k \leq \deg_G(v)$).
			For $i \in \{1, 2, \dotsc, k\}$, let $t_i$ be the token placed at $w_i$, and $L_{w_i}$ be the leg of $G$ containing $w_i$.
			If $\Nei{G}{w_i} \setminus \{v\} \neq \emptyset$ then let $x_i$ be such that $\Nei{G}{w_i} \setminus \{v\} = \{x_i\}$. 
			(Since $G$ is a spider, $w_i$ has at most two neighbors.)
			\begin{figure}[!ht]
				\centering
				\includegraphics[scale=0.6]{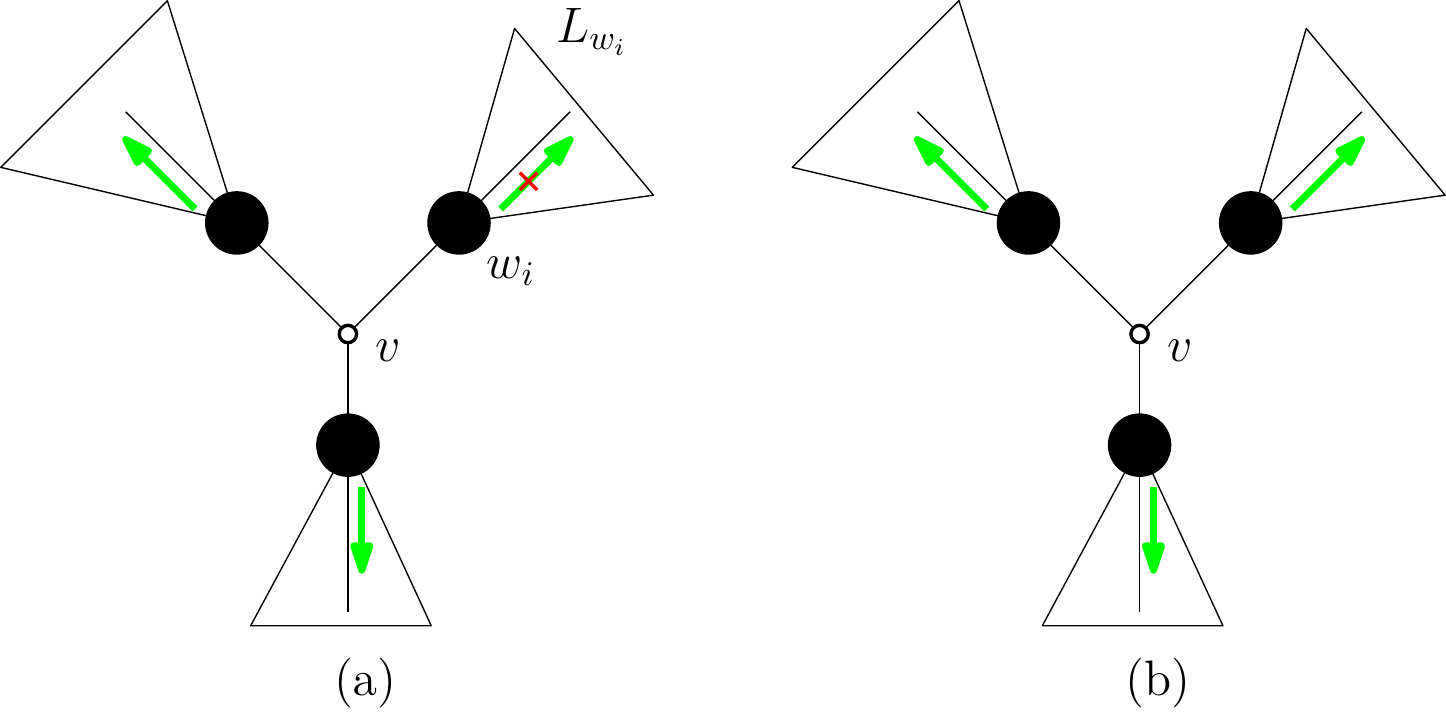}
				\caption{Illustration of \textbf{Case~1} of Lemma~\ref{lem:two-tokens-in-neigh-v}: (a) \textbf{Case~1.1}, and (b) \textbf{Case~1.2}. Here $k = \msize{I \cap \Nei{G}{v}} = 3$, and tokens in $I$ are of black color.}
				\label{fig:two-tokens-in-neigh-v}
			\end{figure}
			\begin{itemize}
				\item \textbf{Case~1.1: There exists $i \in \{1, 2, \dotsc, k\}$ such that $t_i$ is $(L_{w_i}, I_{L_{w_i}})$-rigid.}
				From~\cite[Lemma 2]{DDFEHIOOUY2015}, the token $t_i$ is unique, i.e., there is no $j \neq i$ such that $t_j$ is $(L_{w_j}, I_{L_{w_j}})$-rigid; otherwise, $t_i$ and $t_j$ are both $(G, I)$-rigid, which contradicts our assumption that there are no $(G, I)$-rigid tokens.
				Note that for $j \neq i$, $\Nei{G}{w_j} \setminus \{v\} \neq \emptyset$; otherwise, $t_j$ is clearly $(L_{w_j}, I_{L_{w_j}})$-rigid, which is a contradiction.
				For each $j \in \{1, 2, \dotsc, k\}$ with $j \neq i$, let $\calS_{w_jx_j}$ be the $\sfTS$-sequence of length $\cost(G, I, w_jx_j)$ that moves $t_j$ from $w_j$ to $x_j$, as described in Lemma~\ref{lem:calculate-cost-in-linear-time}.
				(Since $G$ is a spider, such $\calS_{w_jx_j}$ is uniquely determined.)
				Let $\calS_1^i = \bigoplus_{j}\calS_{w_jx_j}$.
				From the proof of Lemma~\ref{lem:no-tokens-in-neigh-v}, $\calS_1^i$ is indeed a $\sfTS$-sequence that moves $t_j$ from $w_j$ to $x_j$, for every $j \neq i$.
				Intuitively, each $\calS_{w_jx_j}$ only moves tokens ``inside'' the leg $L_{w_j}$.
				Let $\Ip$ be the resulting independent set (of performing $\calS_1^i$).
				Clearly, $\max\{\msize{\Ip \cap \Nei{G}{v}}, \msize{J \cap \Nei{G}{v}}\} \leq 1$.
				Let $\calS_2^i$ be the $\sfTS$-sequence that reconfigures $\Ip$ to $J$ as described in Lemma~\ref{lem:one-token-in-neigh-v}.
				We claim that $\calS = \calS_1^i \oplus \calS_2^i$ is a $\sfTS$-sequence between $I$ and $J$ of shortest length.
				It is trivial that $\calS$ is a $\sfTS$-sequence between $I$ and $J$, as it reconfigures $I$ to $\Ip$, and then $\Ip$ to $J$.
				To see that $\calS$ is indeed shortest, note that since $\msize{I_L} \neq \msize{J_L}$ for some leg $L$ of $G$, any $\sfTS$-sequence must move $t_i$ to some vertex not in $\Nei{G}{v}$; otherwise, some token in $I_L$ where $\msize{I_L} > \msize{J_L}$ cannot be moved to its final target vertex.
				Since $t_i$ is $(L_{w_i}, I_{L_{w_i}})$-rigid, the only way is to move $t_i$ ``out of'' $L_{w_i}$.
				Roughly speaking, the token-slides in $\calS_1^i$ is unavoidable, i.e., any $\sfTS$-sequence $\calS$ between $I$ and $J$ contains $\calS_1^i$ as a subsequence.
				Since any token-slide in $\calS$ before $\calS_1^i$ can only be performed ``inside'' a particular leg of $G$, one can assume without loss of generality that $\calS_1^i$ is performed before any other token-slide in $\calS$. 
				Additionally, from Lemma~\ref{lem:one-token-in-neigh-v}, $\calS_2^i$ must be a $\sfTS$-sequence of shortest length between $\Ip$ and $J$.
				Hence, $\calS$ is indeed a $\sfTS$-sequence of shortest length between $I$ and $J$.
				
				\item \textbf{Case~1.2: For every $i \in \{1, 2, \dotsc, k\}$, $t_i$ is not $(L_{w_i}, I_{L_{w_i}})$-rigid.}
				As before, note that $x_i$ exists for every $i$.
				For each $t_i$, using the same technique as in \textbf{Case~1.1}, one can indeed construct a $\sfTS$-sequence $\calS_1^i$ that moves all $t_j$ ($j \neq i$) from $w_j$ to $x_j$ of length $\len(\calS_1^i) = \sum_{j \neq i}\cost(G, I, w_jx_j)$, and a $\sfTS$-sequence $\calS_2^i$ that reconfigures the resulting independent set (after performing $\calS_1^i$) to $J$.
				Let $\calS^i = \calS_1^i \oplus \calS_2^i$.
				Then, $\calS^i$ is indeed a $\sfTS$-sequence that reconfigures $I$ to $J$.
				Let $\calS$ be a $\sfTS$-sequence whose length is smallest among all $\calS^i$.
				We claim that $\calS$ is indeed our desired $\sfTS$-sequence.
				Trivially, $\calS$ reconfigures $I$ to $J$.
				To see that it is indeed shortest, note that since there exists some leg $L$ with $\msize{I_L} \neq \msize{J_L}$, any $\sfTS$-sequence between $I$ and $J$ must perform one of $\calS_1^i$.
				As before, we can also assume without loss of generality that for any $\sfTS$-sequence $\calSp$ between $I$ and $J$, the sequence $\calS_1^i$, if in $\calSp$, is performed before any other token-slide in $\calSp$.
				Then, each $\calS^i$ is of smallest length among all $\sfTS$-sequence $\calSp$ between $I$ and $J$ containing $\calS_1^i$.
				Moreover, it is clear from the construction that if $\calSp$ contains $\calS_1^i$ as a subsequence then it does not contain any $\calS^j_1$ for $j \neq i$.
				Therefore, a $\sfTS$-sequence $\calS$ of smallest length among all $\calS^i$ is indeed our desired $\sfTS$-sequence. 
			\end{itemize}
			
			\item \textbf{Case~2: $\msize{I \cap \Nei{G}{v}} \leq 1$ and $\msize{J \cap \Nei{G}{v}} \geq 2$.}
			
			Analogously to \textbf{Case~1}, one can also construct a $\sfTS$-sequence of shortest length between $I$ and $J$.
			Intuitively, instead of moving tokens in $I \cap \Nei{G}{v}$ (as in \textbf{Case~1}), we now move tokens in $J \cap \Nei{G}{v}$: keep one token fixed, and move all other tokens to their corresponding neighbors (different from $v$).
			Once we have the resulting independent set $\Jp$, the reverse of the above $\sfTS$-sequence can be used to reconfigure $\Jp$ to $J$, and by Lemma~\ref{lem:one-token-in-neigh-v} we already know how to reconfigure $I$ to $\Jp$ using a smallest possible number of tokens.
			Combining these two reconfigurations, we now have a $\sfTS$-sequence that reconfigures $I$ to $J$.
			Our desired $\sfTS$-sequence is the shortest among all (in particular, there are at most $\deg_G(v)$ of them) such $\sfTS$-sequences between $I$ and $J$ above.
			
			\item \textbf{Case~3: $\msize{I \cap \Nei{G}{v}} \geq 2$ and $\msize{J \cap \Nei{G}{v}} \geq 2$.}
			
			A shortest $\sfTS$-sequence between $I$ and $J$ can be constructed by simply combining the techniques in \textbf{Case~1} and \textbf{Case~2}.
		\end{itemize}
		In all above cases, the construction of our desired $\sfTS$-sequence $\calS$ obviously takes $O(n^2)$ time.
		
		We remark that in the described algorithm, $D_G(\calS)$ was not explicitly calculated.
		In the remaining part of this proof, we show how to calculate $D_G(\calS)$.
		It is sufficient to show how to calculate $D_G(\calS)$ in \textbf{Case~1.1}; other cases can be done in similar manner.
		From \textbf{Case~1.1}, $\calS = \calS_1^i \oplus \calS_2^i$.
		We note that from Lemma~\ref{lem:calculate-cost-in-linear-time} $\calS_1^i$ itself does not make detour over any edge of $G$.
		On the other hand, Lemma~\ref{lem:one-token-in-neigh-v} implies that $\calS_2^i$ itself makes detour over at most one edge of $G$ (due to whether Lemma~\ref{lem:one-token-in-neigh-v}(i) holds).
		From Lemma~\ref{lem:lower-bound-len-TS-seq-tree}, it remains to calculate the number of detours made by $\calS_1^i$ and $\calS_2^i$ together.
		From the construction of $\calS_1^i$, note that each move $x \to y$ in $\calS_1^i$ appears exactly once.
		Consider a move $x \to y$ in $\calS_1^i$ such that $(y, x)$ is a directed edge of the corresponding auxiliary graph $A(G, I, J)$.
		By definition, $\msize{I \cap G^x_y} \geq \msize{J \cap G^x_y}$.
		Let $\Ip$ be the resulting independent set after the move $x \to y$.
		Then, it can be shown by induction on the number of such moves in $\calS_1^i$ that $\msize{\Ip \cap G^x_y} > \msize{J \cap G^x_y}$. 
		It follows that at some point, $\calS_2^i$ will have to make a move $y \to x$.
		Together, these moves form detour over $e = xy \in E(G)$.
		Since each move $x \to y$ in $\calS_1^i$ appears exactly once, we must have $D_G(\calS) = D_G(\calS_2^i) + 2\msize{\{\langle x \to y \rangle \in \calS_1^i: (y, x) \in E(A(G, I, J))\}}$.
		In a similar manner, in \textbf{Case~1.2}, $D_G(\calS)$ can be calculated.
		In \textbf{Case~2}, we argue with tokens in $J$ (instead of $I$) and the auxiliary graph $A(G, J, I)$ (instead of $A(G, I, J)$).
		Finally, in \textbf{Case~3}, we simply combine the arguments in \textbf{Cases~1} and \textbf{Case~2}.
	\end{proof}
	
	\section{Conclusion}\label{sec:conclusion}
	
	In this paper, we have shown that one can indeed construct a $\sfTS$-sequence of shortest length between two given independent sets of a spider graph (if exists).
	We hope that our ideas and approaches described here will provide a useful framework for improving the polynomial-time algorithm for {\sc Shortest Sliding Token} for trees~\cite{Sugimori2018a}.
	
	\bibliography{main}

\begin{thebibliography}{26}
\providecommand{\natexlab}[1]{#1}
\providecommand{\url}[1]{\texttt{#1}}
\expandafter\ifx\csname urlstyle\endcsname\relax
  \providecommand{\doi}[1]{doi: #1}\else
  \providecommand{\doi}{doi: \begingroup \urlstyle{rm}\Url}\fi

\bibitem[Belmonte et~al.(2018)Belmonte, Kim, Lampis, Mitsou, Otachi, and
  Sikora]{BelmonteKLMOS18}
R\'{e}my Belmonte, Eun~Jung Kim, Michael Lampis, Valia Mitsou, Yota Otachi, and
  Florian Sikora.
\newblock Token sliding on split graphs.
\newblock \emph{arXiv preprint}, 2018.
\newblock \href{http://arxiv.org/abs/1807.05322}{{\ttfamily arXiv:1807.05322}}.

\bibitem[Bonamy and Bousquet(2017)]{BonamyB17}
Marthe Bonamy and Nicolas Bousquet.
\newblock Token sliding on chordal graphs.
\newblock In \emph{Proceedings of WG 2017}, volume 10520 of \emph{LNCS}, pages
  127--139. Springer, 2017.
\newblock \href{https://doi.org/10.1007/978-3-319-68705-6_10}{\ttfamily\path{
  doi:10.1007/978-3-319-68705-6_10}}.

\bibitem[Bonsma et~al.(2014)Bonsma, Kami{\'n}ski, and
  Wrochna]{BonsmaKaminskiWrochna}
Paul~S. Bonsma, Marcin Kami{\'n}ski, and Marcin Wrochna.
\newblock Reconfiguring independent sets in claw-free graphs.
\newblock In \emph{Proceedings of SWAT 2014}, volume 8503 of \emph{LNCS}, pages
  86--97. Springer, 2014.
\newblock \href{https://doi.org/10.1007/978-3-319-08404-6_8}{\ttfamily\path{
  doi:10.1007/978-3-319-08404-6_8}}.

\bibitem[Demaine and Rudoy(2018)]{DR2017}
Erik~D. Demaine and Mikhail Rudoy.
\newblock A simple proof that the $(n^2-1)$-puzzle is hard.
\newblock \emph{Theoretical Computer Science}, 732:\penalty0 80--84, 2018.
\newblock \href{https://doi.org/10.1016/j.tcs.2018.04.031}{\ttfamily\path{
  doi:10.1016/j.tcs.2018.04.031}}.

\bibitem[Demaine et~al.(2014)Demaine, Demaine, Fox-Epstein, Hoang, Ito, Ono,
  Otachi, Uehara, and Yamada]{Demaine2014-TS-trees}
Erik~D. Demaine, Martin~L. Demaine, Eli Fox-Epstein, Duc~A. Hoang, Takehiro
  Ito, Hirotaka Ono, Yota Otachi, Ryuhei Uehara, and Takeshi Yamada.
\newblock Polynomial-time algorithm for sliding tokens on trees.
\newblock In \emph{Proceedings of ISAAC 2014}, volume 8889 of \emph{LNCS},
  pages 389--400. Springer, 2014.
\newblock \href{https://doi.org/10.1007/978-3-319-13075-0_31}{\ttfamily\path{
  doi:10.1007/978-3-319-13075-0_31}}.

\bibitem[Demaine et~al.(2015)Demaine, Demaine, Fox-Epstein, Hoang, Ito, Ono,
  Otachi, Uehara, and Yamada]{DDFEHIOOUY2015}
Erik~D. Demaine, Martin~L. Demaine, Eli Fox-Epstein, Duc~A. Hoang, Takehiro
  Ito, Hirotaka Ono, Yota Otachi, Ryuhei Uehara, and Takeshi Yamada.
\newblock Linear-time algorithm for sliding tokens on trees.
\newblock \emph{Theoretical Computer Science}, 600:\penalty0 132--142, 2015.
\newblock \href{https://doi.org/10.1016/j.tcs.2015.07.037}{\ttfamily\path{
  doi:10.1016/j.tcs.2015.07.037}}.

\bibitem[Diestel(2010)]{Diestel2010}
Reinhard Diestel.
\newblock \emph{Graph Theory}, volume 173 of \emph{Graduate Texts in
  Mathematics}.
\newblock Springer, 4th edition, 2010.

\bibitem[Fox{-}Epstein et~al.(2015)Fox{-}Epstein, Hoang, Otachi, and
  Uehara]{FoxEpsteinHoangOtachiUehara2015}
Eli Fox{-}Epstein, Duc~A. Hoang, Yota Otachi, and Ryuhei Uehara.
\newblock Sliding token on bipartite permutation graphs.
\newblock In \emph{Proceedings of ISAAC 2015}, volume 9472 of \emph{LNCS},
  pages 237--247. Springer, 2015.
\newblock \href{https://doi.org/10.1007/978-3-662-48971-0_21}{\ttfamily\path{
  doi:10.1007/978-3-662-48971-0_21}}.

\bibitem[Gardner(1964)]{Gardner}
Martin Gardner.
\newblock The hypnotic fascination of sliding-block puzzles.
\newblock \emph{Scientific American}, 210:\penalty0 122--130, 1964.

\bibitem[Gopalan et~al.(2009)Gopalan, Kolaitis, Maneva, and
  Papadimitriou]{Kolaitis}
Parikshit Gopalan, Phokion~G. Kolaitis, Elitza~N. Maneva, and Christos~H.
  Papadimitriou.
\newblock The connectivity of boolean satisfiability: Computational and
  structural dichotomies.
\newblock \emph{SIAM Journal on Computing}, 38\penalty0 (6):\penalty0
  2330--2355, 2009.
\newblock \href{https://doi.org/10.1137/07070440X}{\ttfamily\path{
  doi:10.1137/07070440X}}.

\bibitem[Hearn and Demaine(2005)]{HearnDemaine2005}
Robert~A. Hearn and Erik~D. Demaine.
\newblock {PSPACE}-completeness of sliding-block puzzles and other problems
  through the nondeterministic constraint logic model of computation.
\newblock \emph{Theoretical Computer Science}, 343\penalty0 (1-2):\penalty0
  72--96, 2005.
\newblock \href{https://doi.org/10.1016/j.tcs.2005.05.008}{\ttfamily\path{
  doi:10.1016/j.tcs.2005.05.008}}.

\bibitem[Hearn and Demaine(2009)]{HearnDemaine2009}
Robert~A. Hearn and Erik~D. Demaine.
\newblock \emph{Games, puzzles, and computation}.
\newblock A K Peters, 2009.

\bibitem[Hoang and Uehara(2016)]{HoangUehara2016}
Duc~A. Hoang and Ryuhei Uehara.
\newblock Sliding tokens on a cactus.
\newblock In \emph{Proceedings of ISAAC 2016}, volume~64 of \emph{LIPIcs},
  pages 37:1--37:26. Schloss Dagstuhl - Leibniz-Zentrum fuer Informatik, 2016.
\newblock \href{https://doi.org/10.4230/LIPIcs.ISAAC.2016.37}{\ttfamily\path{
  doi:10.4230/LIPIcs.ISAAC.2016.37}}.

\bibitem[Hoang et~al.(2017)Hoang, Fox{-}Epstein, and
  Uehara]{HoangFoxEpsteinUehara2017}
Duc~A. Hoang, Eli Fox{-}Epstein, and Ryuhei Uehara.
\newblock Sliding tokens on block graphs.
\newblock In \emph{Proceedings of WALCOM 2017}, volume 10167 of \emph{LNCS},
  pages 460--471. Springer, 2017.
\newblock \href{https://doi.org/10.1007/978-3-319-53925-6_36}{\ttfamily\path{
  doi:10.1007/978-3-319-53925-6_36}}.

\bibitem[Ito et~al.(2011)Ito, Demaine, Harvey, Papadimitriou, Sideri, Uehara,
  and Uno]{IDHPSUU}
Takehiro Ito, Erik~D. Demaine, Nicholas J.~A. Harvey, Christos~H.
  Papadimitriou, Martha Sideri, Ryuhei Uehara, and Yushi Uno.
\newblock On the complexity of reconfiguration problems.
\newblock \emph{Theoretical Computer Science}, 412\penalty0 (12-14):\penalty0
  1054--1065, 2011.
\newblock \href{https://doi.org/10.1016/j.tcs.2010.12.005}{\ttfamily\path{
  doi:10.1016/j.tcs.2010.12.005}}.

\bibitem[Kami{\'n}ski et~al.(2012)Kami{\'n}ski, Medvedev, and
  Milani{\v{c}}]{KaminskiMedvedevMilanic2012}
Marcin Kami{\'n}ski, Paul Medvedev, and Martin Milani{\v{c}}.
\newblock Complexity of independent set reconfigurability problems.
\newblock \emph{Theoretical Computer Science}, 439:\penalty0 9--15, 2012.
\newblock \href{https://doi.org/10.1016/j.tcs.2012.03.004}{\ttfamily\path{
  doi:10.1016/j.tcs.2012.03.004}}.

\bibitem[Makino et~al.(2011)Makino, Tamaki, and Yamamoto]{MTY11}
Kazuhisa Makino, Suguru Tamaki, and Masaki Yamamoto.
\newblock An exact algorithm for the boolean connectivity problem for
  $k$-{CNF}.
\newblock \emph{Theoretical Computer Science}, 412\penalty0 (35):\penalty0
  4613--4618, 2011.
\newblock \href{https://doi.org/10.1016/j.tcs.2011.04.041}{\ttfamily\path{
  doi:10.1016/j.tcs.2011.04.041}}.

\bibitem[Mouawad et~al.(2013)Mouawad, Nishimura, Raman, Simjour, and
  Suzuki]{MNRSS13}
Amer~E. Mouawad, Naomi Nishimura, Venkatesh Raman, Narges Simjour, and Akira
  Suzuki.
\newblock On the parameterized complexity of reconfiguration problems.
\newblock In \emph{Proceedings of IPEC 2013}, volume 8246 of \emph{LNCS}, pages
  281--294. Springer, 2013.
\newblock \href{https://doi.org/10.1007/978-3-319-03898-8_24}{\ttfamily\path{
  doi:10.1007/978-3-319-03898-8_24}}.

\bibitem[Mouawad et~al.(2014)Mouawad, Nishimura, Raman, and
  Wrochna]{MouawadNishimuraRamanWrochna}
Amer~E. Mouawad, Naomi Nishimura, Venkatesh Raman, and Marcin Wrochna.
\newblock Reconfiguration over tree decompositions.
\newblock In \emph{Proceedings of IPEC 2014}, volume 8894 of \emph{LNCS}, pages
  246--257. Springer, 2014.
\newblock \href{https://doi.org/10.1007/978-3-319-13524-3_21}{\ttfamily\path{
  doi:10.1007/978-3-319-13524-3_21}}.

\bibitem[Ratner and Warmuth(1990)]{RW90}
D.~Ratner and M.~Warmuth.
\newblock Finding a shortest solution for the {$N \times N$-extension} of the
  $15$-puzzle is intractable.
\newblock \emph{Journal of Symbolic Computation}, 10:\penalty0 111--137, 1990.

\bibitem[Ribeiro and dos Santos(March, 2018)]{ViniCom2018}
Mariana~Teatini Ribeiro and Vin\'{i}cius~Fernandes dos Santos.
\newblock Personal communications, March, 2018.

\bibitem[Slocum and Sonneveld(2006)]{Slocum}
Jerry Slocum and Dic Sonneveld.
\newblock \emph{The 15 puzzle book: how it drove the world crazy}.
\newblock Socum Puzzle Foundations, 2006.

\bibitem[Sugimori(2018)]{Sugimori2018a}
Ken Sugimori.
\newblock Shortest reconfiguration of sliding tokens on a tree, 2018.
\newblock {AAAC} 2018, May, 2018.

\bibitem[Sugimori(May, 2018)]{Sugimori2018b}
Ken Sugimori.
\newblock Personal communications, May, 2018.

\bibitem[van~der Zanden(2015)]{Zanden}
Tom~C. van~der Zanden.
\newblock Parameterized complexity of graph constraint logic.
\newblock In \emph{Proceedings of IPEC 2015}, volume~43 of \emph{LIPIcs}, pages
  282--293. Schloss Dagstuhl - Leibniz-Zentrum fuer Informatik, 2015.
\newblock \href{https://doi.org/10.4230/LIPIcs.IPEC.2015.282}{\ttfamily\path{
  doi:10.4230/LIPIcs.IPEC.2015.282}}.

\bibitem[Yamada and Uehara(2016)]{YamadaUehara2016}
Takeshi Yamada and Ryuhei Uehara.
\newblock Shortest reconfiguration of sliding tokens on a caterpillar.
\newblock In Mohammad Kaykobad and Rossella Petreschi, editors,
  \emph{Proceedings of {WALCOM 2016}}, volume 9627 of \emph{LNCS}, pages
  236--248. Springer, 2016.
\newblock \href{https://doi.org/10.1007/978-3-319-30139-6_19}{\ttfamily\path{
  doi:10.1007/978-3-319-30139-6_19}}.

\end{thebibliography}
\end{document}